\documentclass[12pt,a4paper]{article}

\usepackage{amsfonts}
\usepackage{amsmath}
\usepackage{amssymb}
  
\newtheorem{theorem}{Theorem}

\newtheorem{corollary}[theorem]{Corollary}

\newtheorem{definition}[theorem]{Definition}
\newtheorem{example}[theorem]{Example}

\newtheorem{lemma}[theorem]{Lemma}

\newtheorem{proposition}[theorem]{Proposition}

\newenvironment{proof}[1][Proof]{\noindent\textbf{#1.} }{\ \rule{0.5em}{0.5em}}

\renewcommand{\P}{{\hat P}}

\newcommand{\Ga}{\Gamma}
\newcommand\ld{\lambda}
\newcommand\Si{\Sigma}
\newcommand\De{\Delta}
\newcommand{\sa}{{\rm sa}}
\newcommand{\A}{{\hat A}}
\newcommand{\Hi}{{\cal H}}
\newcommand\BH{\mathcal{B(H)}}
\newcommand\PV{\mathcal{P}(V)}
\newcommand\bra[1]{\langle #1|\,}
\newcommand\ket[1]{\,|#1\rangle}
\newcommand\eq[1]{(\ref{#1})}

\newcommand\SAin[1]{\mbox{``}A\,\varepsilon\,#1\mbox{''}}
\newcommand\Ain[1]{A\,\varepsilon\,#1}
\newcommand\das[1]{\delta(\hat{#1})}            
\newcommand\dastoo[2]{\delta^o(\hat{#2})_{#1}}  
\newcommand\dastoi[2]{\delta^i(\hat{#2})_{#1}}  
\newcommand\dasB[1]{\breve{\delta}(#1)}
\newcommand\spec[1]{{\rm sp}(#1)}
\newcommand\ps[1]{\underline{#1}}        
\newcommand{\Sig}{\ps{\Sigma}}            
\newcommand\Sub[1]{{\rm Sub}(#1)}              
\newcommand\Subcl[1]{{\rm Sub}_{{\rm cl}}(#1)} 
\newcommand\Set{{\bf Sets}}                    
\newcommand\SetC[1]{\Set^{{#1}^{\rm op}}}      
\newcommand\N{\mathcal{N}}
\newcommand\fu[1]{\overline{#1}}							

\newcommand\SetVNop{\SetC{\mathcal{V(N)}}}		

\newcommand\VN{\mathcal{V(N)}}
\newcommand\on[1]{\operatorname{#1}}
\newcommand\Q{\hat Q}
\newcommand\bbC{\mathbb{C}}
\newcommand\bbR{\mathbb{R}}
\newcommand\PN{\mathcal{P}(\N)}
\newcommand\mc[1]{\mathcal{#1}}
\newcommand\Rlr{\ps{\bbR^\leftrightarrow}}
\newcommand\wpsi{\mathfrak{w}^\psi}
\newcommand\dasout[2]{\delta^o(#1)_{#2}}
\newcommand\dasinn[2]{\delta^i(#1)_{#2}}

\begin{document}

\title{The physical interpretation of daseinisation\footnote{Based on joint work with Chris J. Isham.}}
\author{Andreas D\"oring\footnote{andreas.doering@comlab.ox.ac.uk}\\
Computing Laboratory, University of Oxford\\}
\date{April, 2010}
\maketitle

\begin{abstract}
We provide a conceptual discussion and physical interpretation of some of the quite abstract constructions in the topos approach to physics. In particular, the daseinisation process for projection operators and for self-adjoint operators is motivated and explained from a physical point of view. Daseinisation provides the bridge between the standard Hilbert space formalism of quantum theory and the new topos-based approach to quantum theory. As an illustration, we will show all constructions explicitly for a three-dimensional Hilbert space and the spin-$z$ operator of a spin-$1$ particle. This article is a companion to the article by Isham in the same volume \cite{Ish10}.
\end{abstract}

\section{Introduction}	\label{S1}
\paragraph{The topos approach.} The topos approach to quantum theory was initiated by Isham in \cite{Isham97} and Butterfield and Isham in \cite{IB98,IB99,IB00,IB02}. It was developed and broadened into an approach to the formulation of physical theories in general by Isham and the author in \cite{DI(1),DI(2),DI(3),DI(4)}. The long article \cite{DI08} gives a more-or-less exhaustive\footnote{and probably exhausting...} and coherent overview of the approach. More recent developments are the description of arbitrary states by probability measures \cite{Doe08} and further developments \cite{Doe09} concerning the new form of quantum logic that constitutes a central part of the topos approach. For background, motivation and the main ideas see also Isham's article in this volume \cite{Ish10}.

Most of the work so far has been done on standard non-relativistic quantum theory. A quantum system is described by its \emph{algebra of physical quantities}. Often, this can be assumed to be $\BH$, the algebra of all bounded operators on the (separable) Hilbert space $\Hi$ of the system. More generally, one can use a suitable operator algebra. For conceptual and pragmatical reasons, we assume that the algebra of physical quantities is a \emph{von Neumann algebra $\N$} (see e.g. \cite{KR83a,Tak79}). For our purposes, this poses no additional technical difficulty. Physically, it allows the description of quantum systems with symmetry and/or superselection rules. The reader unfamiliar with von Neumann algebras can always assume that $\N=\BH$.

Quantum theory usually is identified with the Hilbert space formalism together with some interpretation. This works fine in a vast number of applications. Moreover, the Hilbert space formalism is very rigid. One cannot just change some part of the structure since this typically brings down the whole edifice. Yet, there are serious conceptual problems with the usual instrumentalist interpretations of quantum theory which become even more severe when one tries to apply quantum theory to gravity and cosmology. For a discussion of some of these conceptual problems, see Isham's article in this volume.

The topos approach provides not merely another interpretation of the Hilbert space formalism, but a mathematical reformulation of quantum theory, based upon structural and conceptual considerations. The resulting formalism is \emph{not} a Hilbert space formalism. The fact that such a reformulation is possible at all is somewhat surprising. (Needless to say, many open questions remain.)

\paragraph{Daseinisation.} In this article, we mainly focus on how the new topos formalism relates to the standard Hilbert space formalism. The main ingredient is the process that we coined \emph{daseinisation}. It will be shown how daseinisation relates familiar structures in quantum theory to structures within the topos associated with a quantum system. Hence, daseinisation gives the `translation' from the ordinary Hilbert space formalism to the topos formalism.

In our presentation, we will keep the use of topos theory to an absolute minimum and do not assume any familiarity with category theory beyond the very basics. Some notions and results from functional analysis, e.g. the spectral theorem, will be used. As an illustration, we will show how all constructions look like concretely for the algebra $\N=\mc B (\bbC^3)$ and the spin-$z$ operator, $\hat S_z\in\mc B (\bbC^3)$, of a spin-$1$ particle. 

In fact, there are two processes called daseinisation. The first is \emph{daseinisation of projections}, which maps projection operators to certain subobjects of the state object. (The state object and its subobjects will be defined below.) This provides the bridge from ordinary Birkhoff-von Neumann quantum logic to a new form of quantum logic that is based upon the internal logic of a topos associated with the quantum system. The second is \emph{daseinisation of self-adjoint operators}, which maps each self-adjoint operator to an arrow in the topos. Mathematically, the two forms of daseinisation are related, but conceptually they are quite different.

The topos approach emphasises the r\^ole of \emph{classical perspectives} onto a quantum system. A classical perspective or \emph{context} is nothing but a set of commuting physical quantities, or more precisely the abelian von Neumann algebra generated by such a set. One of the main ideas is that \emph{all} classical perspectives should be taken into account simultaneously. But given e.g. a projection operator, which represents a proposition about the value of a physical quantity in standard quantum theory, one is immediately faced with the problem that the projection operator is contained in some contexts, but not in all. The idea is to approximate the projection in all those contexts that do not contain it. Likewise, self-adjoint operators, which represent physical quantities, must be approximated suitably to all contexts. Daseinisation is nothing but a process of systematic approximation to all classical contexts.

This article is organised as follows: in section \ref{S2}, we discuss propositions and their representation in classical physics and standard quantum theory. Section \ref{S3} presents some basic structures in the topos approach to quantum theory. In section \ref{S4}, the representation of propositions in the topos approach is discussed -- this is via daseinisation of projections. Section \ref{S5}, which is the longest and most technical, shows how physical quantities are represented in the topos approach; here we present daseinisation of self-adjoint operators is discussed. Section \ref{S6} concludes.

\section{\mbox{Propositions and their representation in}\\ classical physics and standard quantum\\ theory}	\label{S2}
\paragraph{Propositions.} Let $S$ be some physical system, and let $A$ denote a physical quantity pertaining to the system $S$ (e.g. position, momentum, energy, angular momentum, spin...). We are concerned with propositions of the form ``the physical quantity $A$ has a value in the set $\De$ of real numbers'', for which we use the shorthand notation ``$\Ain{\De}$''. In all applications, $\De$ will be a Borel subset of the real numbers $\bbR$.

Arguably, physics is fundamentally about what we can say about the truth values of propositions of the form ``$\Ain{\De}$'' (where $A$ varies over the physical quantities of the system and $\De$ varies over the Borel subsets of $\bbR$) when the system is in a given state. Of course, the state may change in time and hence the truth values also will change in time.\footnote{One could also think of propositions about \emph{histories} of a physical system $S$ such that propositions refer to multiple times, and the assignment of truth values to such propositions (in a given, evolving state). For some very interesting recent results on a topos formulation of histories see Flori's recent article \cite{Flo08}.}

Speaking about propositions like ``$\Ain{\De}$'' requires a certain amount of conceptualising. We accept that it is sensible to talk about physical systems as suitably separated entities, that each such physical system is characterised by its physical quantities, and that the range of values of a physical quantity $A$ is a subset of the real numbers. If we want to assign truth values to propositions, then we need the concept of a state of a physical system $S$, and the truth values of propositions will depend on the state. If we regard these concepts as natural and basic (and potentially as prerequisites for doing physics at all), then a proposition ``$\Ain{\De}$'' \emph{refers to the world} in the most direct conceivable sense.\footnote{Of course, as explained by Isham in his article \cite{Ish10}, it is one of the central motivations for the whole topos programme to find a mathematical framework for physical theories that is not depending fundamentally on the real numbers. In particular, the premise that physical quantities have real values is doubtful in this light. How this seeming dilemma is solved in the topos approach to quantum theory will become clear later in section \ref{Sec_RepOfPhysQuants}.} We will deviate from common practice, though, by insisting that also in quantum theory, a proposition is about `how things are', and is not to be understood as a counterfactual statement, i.e., it is not merely about what we would obtain as a measurement result if we were to perform a measurement.

\paragraph{Classical physics, state spaces and realism.} In classical physics, a state of the system $S$ is represented by an element of a set, namely by a point of the state space $\mathcal{S}$ of the system.\footnote{We avoid the usual synonym `phase space', which seems to be a historical misnomer. In any case, there are no phases in a phase space.} Each physical quantity $A$ is represented by a real-valued function $f_A:\mathcal{S}\rightarrow\bbR$ on the state space. A proposition ``$\Ain{\De}$'' is represented by a certain subset of the state space, namely the set $f_A^{-1}(\De)$. We assume that $f_A$ is (at least) measurable. Since $\De$ is a Borel set, the subset $f_A^{-1}(\De)$ of state space representing ``$\Ain{\De}$'' is a Borel set, too.

In any state, each physical quantity $A$ has a value, which is simply given by evaluation of the function $f_A$ representing $A$ at the point $s\in\mc{S}$ of state space representing the state, i.e., the value is $f_A(s)$. Every proposition ``$\Ain{\De}$'' has a Boolean truth value in each given state $s\in\mc{S}$:
\begin{equation}
			v(\SAin{\De};s)=\left\{
			\begin{tabular}
						[c]{ll}%
						$true$ & if $s\in f_ A^{-1}(\De)$\\
						$false$ & if $s\notin f_A^{-1}(\De)$.
			\end{tabular}
			\ \right.
\end{equation}
Let ``$\Ain{\De}$'' and ``$B\varepsilon\Gamma$'' be two different propositions, represented by the Borel subsets $f_A^{-1}(\De)$ and $f_B^{-1}(\Gamma)$ of $\mc{S}$, respectively. The union $f_A^{-1}(\De)\cup f_B^{-1}(\Gamma)$ of the two Borel subsets is another Borel subset, and this subset represents the proposition ``$\Ain{\De}$ or $B\varepsilon\Gamma$'' (disjunction). Similarly, the intersection $f_A^{-1}(\De)\cap f_B^{-1}(\Gamma)$ is a Borel subset of $\mc{S}$, and this subset represents the proposition ``$\Ain{\De}$ and $B\varepsilon\Gamma$'' (conjunction). Both conjunction and disjunction can be extended to arbitrary countable families of Borel sets representing propositions. The set-theoretic operations of taking unions and intersections distribute over each other. Moreover, the negation of a proposition ``$\Ain{\De}$'' is represented by the complement $\mc{S}\backslash f_A^{-1}(\De)$ of the Borel set $f_A^{-1}(\De)$. Clearly, the empty subset $\emptyset$ of $\mc{S}$ represents the trivially false proposition, while the maximal Borel subset, $\mc{S}$ itself, represents the trivially true proposition. The set $\mc{B(S)}$ of Borel subsets of the state space $\mc{S}$ of a classical system thus is a Boolean $\sigma$-algebra, i.e., a $\sigma$-complete distributive lattice with complement. Stone's theorem shows that every Boolean algebra is isomorphic to the algebra of subsets of a suitable space, so Boolean logic is closely tied to the use of sets.

The fact that in a given state $s$ each physical quantity has a value and each proposition has a truth value makes classical physics a \emph{realist} theory.\footnote{Here, we could enter into an interesting, but potentially never-ending discussion on physical reality, ontology and epistemology etc. We avoid that and posit that for the purpose of this paper, a realist theory is one that, in any given state, allows to assign truth values to all propositions of the form ``$\Ain{\De}$''. Moreover, we require that there is a suitable logical structure, in particular a deductive system, in which we can argue about (representatives of) propositions.}

\paragraph{Representation of propositions in standard quantum theory.} In contrast to that, in quantum physics it is not possible to assign real values to all physical quantities at once. This is the content of the Kochen-Specker theorem. In the standard Hilbert space formulation of quantum theory, each physical quantity $A$ is represented by a self-adjoint operator $\A$ on some Hilbert space $\Hi$. The range of possible real values that $A$ can take is given by the spectrum $\spec{\A}$ of the operator $\A$. Of course, there is a notion of states in quantum theory: in the simplest version, they are given by unit vectors in the Hilbert space $\Hi$. There is a particular mapping taking self-adjoint operators and unit vectors to real numbers, namely the evaluation
\begin{equation}
			(\A,\ket\psi)\longmapsto\bra\psi\A\ket\psi.
\end{equation}
In general (unless $\ket\psi$ is an eigenstate of $\A$), this real value is \emph{not} the value of the physical quantity $A$ in the state described by $\psi$. It rather is the expectation value, which is a statistical and instrumentalist notion. The physical interpretation of the mathematical formalism of quantum theory fundamentally depends on measurements and observers.

According to the spectral theorem, propositions like ``$\Ain{\De}$'' are represented by projection operators $\P=\hat E[\Ain{\De}]$. Each projection $\P$ corresponds to a unique closed subspace $U_\P$ of the Hilbert space $\Hi$ of the system and vice versa. The intersection of two closed subspaces is a closed subspace, which can be taken as the definition of a conjunction. The closure of the subspace generated by two closed subspaces is a candidate for the disjunction, and the function sending a projection $\P$ to $\hat 1-\P$ is an orthocomplement. Birkhoff and von Neumann suggested in their seminal paper \cite{BvN36} to interpret these mathematical operations as providing a logic for quantum systems.

At first sight, this looks similar to the classical case: the Hilbert space $\Hi$ now takes the r\^ole of a state space, while its closed subspaces represent propositions. But there is an immediate, severe problem: if the Hilbert space $\Hi$ is at least two-dimensional, then the lattice $\mc{L(H)}$ of closed subspaces is \emph{non-distributive} (and so is the isomorphic lattice of projections). This makes it very hard to find a proper semantics of quantum logic.\footnote{We cannot discuss the merits and shortcomings of quantum logic here. Suffice it to say that there are more conceptual and interpretational problems, and a large number of developments and abstractions from standard quantum logic, i.e., the lattice of closed subspaces of Hilbert space. An excellent review is \cite{DCG02}.}

\section{Basic structures in the topos approach to quantum theory}	\label{S3}
\subsection{Contexts}
As argued by Isham in this volume, the topos formulation of quantum theory is based upon the idea that one takes the collection of \emph{all} classical perspectives on a quantum system. No single classical perspective can deliver a complete picture of the quantum system, but the collection of all of them may well do. As we will see, it is also of great importance how the classical perspectives relate to each other.

These ideas are formalised in the following way: we consider a quantum system as being described by a \emph{von Neumann algebra} $\N$ (see e.g. \cite{KR83a}). Such an algebra is always given as a subalgebra of $\BH$, the algebra of all bounded operators on a suitable Hilbert space $\Hi$. We can always assume that the identity in $\N$ is the identity operator $\hat 1$ on $\Hi$. Since von Neumann algebras are much more general than just the algebras of the form $\BH$, and since all our constructions work for arbitrary von Neumann algebras, we present the results in this general form. For a very good introduction to the use of operator algebras in quantum theory see \cite{Emch84}. Von Neumann algebras can be used to describe quantum systems with symmetry and/or superselection rules. Throughout, we will use the very simple example of $\N=\mc{B}(\bbC^3)$ to illustrate the constructions.

We assume that the self-adjoint operators in the von Neumann algebra $\N$ associated with our quantum system represent the \emph{physical quantities} of the system. Since a quantum system has non-commuting physical quantities, the algebra $\N$ is \emph{non-abelian}. The $\bbR$-vector space of self-adjoint operators in $\N$ is denoted as $\N_\sa$.

A classical perspective is given by a collection of commuting physical quantities. Such a collection determines an \emph{abelian subalgebra}, typically denoted as $V$, of the non-abelian von Neumann algebra $\N$. An abelian subalgebra is often called a \emph{context}, and we will use the notions classical perspective, context and abelian subalgebra synonymously. We will only consider abelian subalgebras that
\begin{itemize}
	\item[(a)] are von Neumann algebras, i.e., they are closed in the weak topology. The technical advantage is that the spectral theorem holds in both $\N$ and its abelian von Neumann subalgebras, and that the lattices of projections in $\N$ and its abelian von Neumann subalgebras are complete;
	\item[(b)] contain the identity operator $\hat 1$.
\end{itemize}
Given a von Neumann algebra $\N$, let $\VN$ be the set of all its abelian von Neumann subalgebras which contain $\hat 1$. By convention, the trivial abelian subalgebra $V_0=\bbC\hat 1$ is not contained in $\VN$. If some subalgebra $V'\in\VN$ is contained in a larger subalgebra $V\in\VN$, then we denote the inclusion as $i_{V'V}:V'\rightarrow V$. Clearly, $\VN$ is a partially ordered set under inclusion, and as such is a simple kind of \emph{category} (see e.g. \cite{McL98}). The objects in this category are the abelian von Neumann subalgebras, and the arrows are the inclusions. Clearly, there is at most one arrow between any two objects. We call $\VN$ the \emph{category of contexts} of the quantum system described by the von Neumann algebra $\N$.

Considering the abelian parts of a non-abelian structure may seem trivial, yet in fact it is not, since the context category $\VN$ keeps track of the relations between contexts: whenever two abelian subalgebras $V,\tilde V$ have a non-trivial intersection, then there are inclusion arrows
\begin{equation}
			V\longleftarrow V\cap\tilde V\longrightarrow\tilde V
\end{equation}
in $\VN$. Every self-adjoint operator $\A\in V\cap\tilde V$ can be written as $g(\hat B)$ for some $\hat B\in V_{\sa}$, where $g:\bbR\rightarrow\bbR$ is a Borel function, and, at the same time, as another $h(\hat C)$ for some $\hat C\in\tilde V_\sa$ and some other Borel function $h$.\footnote{This follows from the fact that each abelian von Neumann algebra $V$ is generated by a single self-adjoint operator, see e.g. Prop. 1.21 in \cite{Tak79}. A Borel function $g:\bbR\rightarrow\bbC$ takes a self-adjoint operator $\A$ in a von Neumann algebra to another operator $g(\A)$ in the same algebra. If $g$ is real-valued, then $g(\A)$ is self-adjoint. For details, see \cite{KR83a,Tak79}.} The point is that while $\A$ commutes with $\hat B$ and $\A$ commutes with $\hat C$, the operators $\hat B$ and $\hat C$ do not necessarily commute. In that way, the context category $\VN$ encodes a lot of information about the algebraic structure of $\N$, not just between commuting operators, and about the relations between contexts.

If $V'\subset V$, then the context $V'$ contains less self-adjoint operators and less projections than the context $V$, so one can describe less physics from the perspective of $V'$ than from the perspective of $V$. The step from $V$ to $V'$ hence involves a suitable kind of \emph{coarse-graining}. We will see later how daseinisation of projections resp. self-adjoint operators implements this informal idea of coarse-graining.

\begin{example}			\label{Ex1}
Let $\Hi=\bbC^3$, and let $\N=\mc{B}(\bbC^3)$, the algebra of all bounded linear operators on $\bbC^3$. $\mc{B}(\bbC^3)$ is the algebra $M_3(\bbC)$ of $3\times 3$-matrices with complex entries, acting as linear transformations on $\bbC^3$. Let $(\psi_1,\psi_2,\psi_3)$ be an orthonormal basis of $\bbC^3$, and let $(\P_1,\P_2,\P_3)$ be the three projections onto the one-dimensional subspaces $\bbC\psi_1,\bbC\psi_2$ and $\bbC\psi_3$, respectively. Clearly, the projections $\P_1,\P_2,\P_3$ are pairwise orthogonal, i.e., $\P_i\P_j=\delta_{ij}\P_i$.

There is an abelian subalgebra $V$ of $\mc{B}(\bbC^3)$ generated by the three projections. One can use von Neumann's double commutant construction and define $V=\{\P_1,\P_2,\P_3\}''$ (see \cite{KR83a}). More concretely, $V=\rm{lin}_{\bbC}(\P_1,\P_2,\P_3)$. Even more explicitly, one can pick the matrix representation of the projections $\P_i$ such that
\begin{equation}
			\P_1=
			\left(\begin{array}
						[c]{ccc}%
						1 & 0 & 0\\
						0 & 0 & 0\\
						0 & 0 & 0
			\end{array}\right),\ \ \
			\P_2=
			\left(\begin{array}
						[c]{ccc}%
						0 & 0 & 0\\
						0 & 1 & 0\\
						0 & 0 & 0
			\end{array}\right),\ \ \
			\P_3=
			\left(\begin{array}
						[c]{ccc}%
						0 & 0 & 0\\
						0 & 0 & 0\\
						0 & 0 & 1
			\end{array}\right).			
\end{equation}
The abelian algebra $V$ generated by these projections then consists of all diagonal $3\times 3$-matrices (with complex entries on the diagonal). Clearly, there is no larger abelian subalgebra of $\mc{B}(\bbC^3)$ that contains $V$, so $V$ is maximal abelian.

Every orthonormal basis $(\tilde\psi_1,\tilde\psi_2,\tilde\psi_3)$ determines a maximal abelian subalgebra $\tilde V$ of $\mc{B}(\bbC^3)$. There is some redundancy, though: if two bases differ only by a permutation and/or phase factors of the basis vectors, then they generate the same maximal abelian subalgebra. Evidently, there are uncountably many maximal abelian subalgebras of $\mc{B}(\bbC^3)$.

There also are non-maximal abelian subalgebras. Consider for example the algebra generated by $\P_1$ and $\P_2+\P_3$. This algebra is denoted as $V_{\P_1}$ and is given as $V_{\P_1}=\rm{lin}_\bbC(\P_1,\P_2+\P_3)=\bbC\P_1+\bbC(\hat 1-\P_1)$. Here, we use that $\P_2+\P_3=\hat 1-\P_1$. There are uncountably many non-maximal subalgebras of $\mc{B}(\bbC^3)$.

The trivial projections $\hat 0$ and $\hat 1$ are contained in every abelian subalgebra. An algebra of the form $V_{\P_1}$ contains two non-trivial projections: $\P_1$ and $\hat 1-\P_1$. There is the trivial abelian subalgebra $V_0=\bbC\hat 1$ consisting of multiples of the identity operator only. By convention, we will not consider the trivial algebra (it is not included in the partially ordered set $\VN$).

The set $\VN=\mc{V}(\mc{B}(\bbC^3))=:\mc{V}(\bbC^3)$ of abelian subalgebras of $\mc{B}(\bbC^3)$ hence can be divided into two subsets: the maximal abelian subalgebras, corresponding to orthonormal bases as described, and the ones of the form $V_{\P_1}$ that are generated by a single projection.\footnote{In other (larger) algebras than $\mc{B}(\bbC^3)$, the situation is more complicated, of course.} Clearly, given two different maximal abelian subalgebras $V$ and $\tilde V$, neither one contains the other. Similarly, if we have two non-maximal algebras $V_{\P_1}$ and $V_{\P_2}$, neither one will contain the other, with one exception: if $\P_2=\hat 1-\P_1$, then $V_{\P_1}=V_{\P_2}$. This comes from the more general fact that if a unital operator algebra contains an operator $\A$, then it also contains the operator $\hat 1-\A$.

Of course, non-maximal subalgebras can be contained in maximal ones. Let $V$ be the maximal abelian subalgebra generated by three pairwise orthogonal projections $\P_1,\P_2,\P_3$. Then $V$ contains three non-maximal abelian subalgebras, namely $V_{\P_1},V_{\P_2}$ and $V_{\P_3}$. Importantly, each non-maximal abelian subalgebra is contained in many different maximal ones: consider two projections $\Q_2,\Q_3$ onto one-dimensional subspaces in $\bbC^3$ such that $\P_1,\Q_2,\Q_3$ are pairwise orthogonal. We assume that $\Q_2,\Q_3$ are such that the maximal abelian subalgebra $\tilde V$ generated by $\P_1,\Q_2,\Q_3$ is different from the algebra $V$ generated by the three pairwise orthogonal projections $\P_1,\P_2,\P_3$. Then both $V$ and $\tilde V$ contain $V_{\P_1}$ as a subalgebra. This argument makes clear that $V_{\P_1}$ is contained in every maximal abelian subalgebra that contains the projection $\P_1$.

We take inclusion of smaller into larger abelian subalgebras (here: of non-maximal into maximal ones) as a partial order on the set $\mc{V}(\bbC^3)$.

Since $\bbC^3$ is finite-dimensional, the algebra $\mc{B}(\bbC^3)$ is both a von Neumann and a $C^*$-algebra. The same holds for the abelian subalgebras. We do not have to worry about weak closedness. On infinite-dimensional Hilbert spaces, these questions become important. In general, an abelian von Neumann subalgebra $V$ is generated by the collection $\PV$ of its projections (which all commute with each other, of course) via the double commutant construction.
\end{example}

\subsection{Gel'fand spectra and the spectral presheaf}
\paragraph{Introduction.} We will now make use of the fact that each abelian ($C^*$-) algebra $V$ of operators can be seen as an algebra of continuous functions on a suitable topological space, the \emph{Gel'fand spectrum} of $V$. At least locally, for each context, we thus obtain a mathematical formulation of quantum theory that is similar to classical physics, with a state space and physical quantities as real-valued functions on this space. The `local state spaces' are then combined into one global structure, the \emph{spectral presheaf}, which serves as a state space analogue for quantum theory.

\paragraph{Gel'fand spectra.} Let $V\in\VN$ be an abelian subalgebra of $\N$, and let $\Si_V$ denote its Gel'fand spectrum. $\Si_V$ consists of the multiplicative states on $V$, i.e., the positive linear functionals $\ld:V\longrightarrow\bbC$ of norm $1$ such that
\begin{equation}
			\forall\A,\hat B\in V:\ld(\A\hat B)=\ld(\A)\ld(\hat B).
\end{equation}
The elements $\ld$ of $\Si_V$ also are pure states of $V$ and algebra homomorphisms from $V$ to $\bbC$. (It is useful to have these different aspects in mind.) The set $\Si_V$ is a compact Hausdorff space in the weak$^*$ topology \cite{KR83a,Tak79}.

Let $\ld\in\Si_V$ be given. It can be shown that for any self-adjoint operator $\A\in V_\sa$, it holds that
\begin{equation}
			\ld(\A)\in\spec{\A},
\end{equation}
and every element $s\in\spec{\A}$ is given as $s=\tilde\ld(\A)$ for some $\tilde\ld\in\Si_V$. Moreover, if $g:\bbR\rightarrow\bbR$ is a Borel function, then
\begin{equation}			\label{FUNC}
			\ld(g(\A)=g(\ld(\A)).
\end{equation}
This implies that an element $\ld$ of the Gel'fand spectrum $\Si_V$ can be seen as a \emph{valuation} on $V_\sa$, i.e., a function sending each self-adjoint operator in $V$ to some element of its spectrum in such a way that equation \eq{FUNC}, the FUNC principle, holds. Within each context $V$, all physical quantities in $V$ (and only those!) can be assigned values at once, and the assignment of values commutes with taking (Borel) functions of the operators. Every element $\ld$ of the Gel'fand spectrum gives a different valuation on $V_\sa$.

The well-known \emph{Gel'fand representation theorem} shows that the abelian von Neumann algebra $V$ is isometrically $*$-isomorphic (i.e., isomorphic as a $C^*$-algebra) to the algebra of continuous, complex-valued functions on its Gel'fand spectrum $\Si_V$.\footnote{The Gel'fand representation theorem holds for abelian $C^*$-algebras. Since every von Neumann algebra is a $C^*$-algebra, the theorem applies to our situation.} The isomorphism is given by
\begin{eqnarray}
			V &\longrightarrow& C(\Si_V)\\
			\A &\longmapsto& \overline{A},
\end{eqnarray}
where $\overline{A}(\ld):=\ld(\A)$ for all $\ld\in\Si_V$. The function $\overline{A}$ is called the \emph{Gel'fand transform} of the operator $\A$. If $\A$ is a self-adjoint operator, then $\overline{A}$ is real-valued, and $\overline{A}(\Si_V)=\spec{\A}$.

The fact that the self-adjoint operators in an abelian (sub)algebra $V$ can be written as real-valued functions on the compact Hausdorff space $\Si_V$ means that the Gel'fand spectrum $\Si_V$ plays a r\^ole exactly like the state space of a classical system. Of course, since $V$ is abelian, not all self-adjoint operators $\A\in\N_{\sa}$ representing physical quantities of the quantum system are contained in $V$. We will interpret the Gel'fand spectrum $\Si_V$ as a \emph{local state space} at $V$. Local here means `at the abelian part $V$ of the global non-abelian algebra $\N$'.

For a projection $\P\in\PV$, we have
\begin{equation}
			\ld(\P)=\ld(\P^{2})=\ld(\P)\ld(\P),
\end{equation}
so $\ld(\P)\in\{0,1\}$ for all projections. This implies that the Gel'fand transform $\overline{P}$ of $\P$ is the characteristic function of some subset of $\Si_V$.\footnote{The fact that a characteristic function can be continuous shows that the Gel'fand topology on $\Si_V$ is pretty wild: the Gel'fand spectrum of an abelian von Neumann algebra is \emph{extremely disconnected}, that is, the closure of each open set is open.} Let $A$ be a physical quantity that is represented by some self-adjoint operator in $V$, and let ``$\Ain{\De}$'' be a proposition about the value of $A$. Then we know from the spectral theorem that there is a projection $\hat E[\Ain{\De}]$ in $V$ that represents the proposition. We saw that $\ld(\hat E[\Ain{\De}])\in\{0,1\}$, and by interpreting $1$ as `true' and $0$ as `false', we see that a valuation $\ld\in\Si_V$ assigns a Boolean truth-value to each proposition ``$\Ain{\De}$''.

The precise relation between projections in $V$ and subsets of the Gel'fand spectrum of $V$ is as follows: let $\P\in\PV$, and define
\begin{equation}
				S_\P:=\{\ld\in\Si_V\mid\ld(\P)=1\}.
\end{equation}
One can show that $S_\P\subseteq\Si_V$ is \emph{clopen}, i.e., closed and open in $\Si_V$. Conversely, every clopen subset $S$ of $\Si_V$ determines a unique projection $\P_S\in\PV$, given as the inverse Gel'fand transform of the characteristic function of $S$. There is a lattice isomorphism
\begin{eqnarray}			\label{alpha}
			\alpha:\PV &\longrightarrow& \mc{C}l(\Si_V)\\			\nonumber
			\P &\longmapsto& S_\P
\end{eqnarray}
between the lattice of projections in $V$ and the lattice of clopen subsets of $\Si_V$. Thus, starting from a proposition ``$\Ain{\De}$'', we obtain a clopen subset $S_{\hat E[\Ain{\De}]}$ of $\Si_V$. This subset consists of all valuations, i.e., pure states $\ld$ of $V$ such that $\ld(\hat E[\Ain\De])=1$, which means that the proposition ``$\Ain\De$'' is true in the state/under the valuation $\ld$. This strengthens the interpretation of $\Si_V$ as a local state space.

We saw that each context $V$ provides us with a local state space, and one of the main ideas in the work by Isham and Butterfield \cite{IB98,IB99,IB00,IB02} was to form a single global object from all these local state spaces.

In order to keep track of the inclusion relations between the abelian subalgebras, we must relate the Gel'fand spectra of larger and smaller subalgebras in a natural way. Let $V,V'\in\VN$ such that $V'\subseteq V$. Then there is a mapping
\begin{eqnarray}			\label{Restriction}
			r:\Si_V &\longrightarrow& \Si_{V'}\\			 \nonumber
			\ld &\longmapsto& \ld|_{V'},
\end{eqnarray}
sending each element $\ld$ of the Gel'fand spectrum of the larger algebra $V$ to its restriction $\ld|_{V'}$ to the smaller algebra $V'$. It is well-known that the mapping $r:\Si_V\rightarrow\Si_{V'}$ is continuous with respect to the Gel'fand topologies and surjective. Physically, this means that every valuation $\ld'$ on the smaller algebra is given as the restriction of a valuation on the larger algebra.

\paragraph{The spectral presheaf.} We can now define the central object in the topos approach to quantum theory: the \emph{spectral presheaf} $\Sig$. A \emph{presheaf} is a contravariant, $\Set$-valued functor, see e.g. \cite{McL98}.\footnote{We will only make minimal use of category theory in this article. In particular, the following definition should be understandable in itself, without further knowledge of functors etc.} Our base category, the domain of this functor, is the context category $\VN$. To each object in $\VN$, i.e., to each context $V$, we assign its Gel'fand spectrum $\Sig_V:=\Si_V$, and to each arrow in $\VN$, i.e., each inclusion $i_{V'V}$, we assign a function from the set $\Sig_V$ to the set $\Sig_{V'}$. This function is $\Sig(i_{V'V}):=r$ (see equation \eq{Restriction}). It implements the concept of coarse-graining on the level of the local state spaces (Gel'fand spectra).

The spectral presheaf $\Sig$ associated with the von Neumann algebra $\N$ of a quantum system is the analogue of the state space of a classical system. The spectral presheaf is not a set, and hence not a space in the usual sense. It rather is a particular, $\Set$-valued functor, built from the Gel'fand spectra of the abelian subalgebras of $\N$, and the canonical restriction functions between them.

We come back to our example, $\N=\mc{B}(\bbC^3)$, and describe its spectral presheaf:

\begin{example}			\label{Ex_SigOfB(C3)}
Let $\N=\mc{B}(\bbC^3)$ as in Ex. \ref{Ex1}, and let $V\in\mc{V}(\bbC^3)$ be a maximal abelian subalgebra. As discussed above, $V$ is of the form $V=\{\P_1,\P_2,\P_3\}''$ for three pairwise orthogonal rank-$1$ projections $\P_1,\P_2,\P_3$. The Gel'fand spectrum $\Si_V$ of $V$ has three elements (and is equipped with the discrete topology, of course). The spectral elements are given as
\begin{equation}
			\ld_i(\P_j)=\delta_{ij}\ \ \ \ (i=1,2,3).
\end{equation}
Let $\A\in V$ be an arbitrary operator. Then $\A=\sum_{i=1}^3 a_i\P_i$ for some (unique) complex coefficients $a_i$. We have $\ld_i(\A)=a_i$. If $\A$ is self-adjoint, then the $a_i$ are real and are the eigenvalues of $\A$.

The non-maximal abelian subalgebra $V_{\P_1}$ has two elements in its Gel'fand spectrum. Let us call them $\ld'_1$ and $\ld'_2$, where $\ld'_1(\P)=1$ and $\ld'_1(\hat 1-\P)=0$, while $\ld'_2(\P)=0$ and $\ld'_2(\hat 1-\P)=1$. The restriction mapping $r:\Si_V\rightarrow\Si_{V'}$ sends $\ld_1$ to $\ld'_1$. Both $\ld_2$ and $\ld_3$ are mapped to $\ld'_2$.

Analogous relations hold for the other non-maximal abelian subalgebras $V_{\P_2}$ and $V_{\P_3}$ of $V$. Since the spectral presheaf is given by assigning to each abelian subalgebra its Gel'fand spectrum, and to each inclusion the corresponding restriction function between the Gel'fand spectra, we have a complete, explicit description of the spectral presheaf $\Sig$ of the von Neumann algebra $\mc{B}(\bbC^3)$.

The generalisation to higher dimensions is straightforward. In particular, if $\rm{dim}(\Hi)=n$ and $V$ is a maximal abelian subalgebra of $\mc{B}(\bbC^n)$, then then Gel'fand spectrum consists of $n$ elements. Topologically, the Gel'fand spectrum $\Si_V$ is a discrete space. (For infinite-dimensional Hilbert spaces, where more complicated von Neumann algebras than $\mc{B}(\bbC^n)$ exist, this is not true in general.)
\end{example}

\section{Representation of propositions in the topos approach -- Daseinisation of projections}	\label{S4}
We now consider the representation of propositions like ``$\Ain{\De}$'' in our topos scheme. We already saw that the spectral presheaf $\Sig$ is an analogue of the state space $\mc{S}$ of a classical system. In analogy to classical physics, where propositions are represented by (Borel) subsets of state space, in the topos approach propositions will be represented by suitable subobjects of the spectral presheaf.

\paragraph{Coarse-graining of propositions.} Consider a proposition ``$\Ain{\De}$'', where $A$ is some physical quantity of the quantum system under consideration. We assume that there is a self-adjoint operator $\A$ in the von Neumann algebra $\N$ of the system that represents $A$. From the spectral theorem, we know that there is a projection operator $\P:=\hat E[\Ain{\De}]$ that represents the proposition ``$\Ain{\De}$''. As is well-known, the mapping from propositions to projections is many-to-one. One can form equivalence classes of propositions to obtain a bijection. In a slight abuse of language, we will usually refer to \emph{the} proposition represented by a projection. Since $\N$ is a von Neumann algebra, the spectral theorem holds. In particular, for all propositions ``$\Ain\De$'', we have $\P=\hat E[\Ain\De]\in\PN$. The projections representing propositions are all contained in the von Neumann algebra. If we had chosen a more general $C^*$-algebra instead, this would not have been the case.

We stipulate that in the topos formulation of quantum theory all classical perspectives/contexts must be taken into account at the same time, so we have to adapt the proposition ``$\Ain{\De}$'' resp. its representing projection $\P$ to all contexts. The idea is very simple: in every context $V$, we pick the strongest proposition implied by ``$\Ain{\De}$'' that can be made from the perspective of this context. On the level of projections, this amounts to taking the \emph{smallest} projection in any context $V$ that is larger than or equal to $\P$.\footnote{Here, we use the interpretation that $\P<\Q$ for two projections $\P,\Q\in\PN$ means that the proposition represented by $\P$ implies the proposition represented by $\Q$. This is customary in quantum logic, and we use it to motivate our construction, though the partial order on projections will \emph{not} be the implication relation for the form of quantum logic resulting from our scheme.} We write, for all $V\in\VN$,
\begin{equation}			\label{Eq_dastooVP}
			\dastoo{V}{P}:=\bigwedge\{\Q\in\PV\mid\Q\geq\P\}
\end{equation}
for this approximation of $\P$ to $V$. If $\P\in\PV$, then clearly the approximation will give $\P$ itself, i.e. $\dastoo{V}{P}=\P$. If $\P\notin\PV$, then $\dastoo{V}{P}>\P$. For many contexts $V$, it holds that $\dastoo{V}{P}=\hat 1$.

Since the projection $\dastoo{V}{P}$ lies in $V$, it corresponds to a proposition about some physical quantity described by a self-adjoint operator \emph{in} $V$. If $\dastoo{V}{P}$ happens to be a spectral projection of the operator $\A$, then the proposition corresponding to $\dastoo{V}{P}$ is of the form ``$\Ain\Ga$'' for some Borel set $\Ga$ of real numbers that is larger than the set $\De$ in the original proposition ``$\Ain\De$''.\footnote{We remark that in order to have $\dastoo{V}{P}\in\PV$, it is sufficient, but \emph{not} necessary that $\A\in V_\sa$.} The fact that $\Ga\supset\De$ shows that the mapping
\begin{eqnarray}
			\delta^o_V:\PN &\longrightarrow& \PV\\			\nonumber
			\P &\longmapsto& \dastoo{V}{P}
\end{eqnarray}
implements the idea of coarse-graining on the level of projections (resp. the corresponding propositions).

If $\dastoo{V}{P}$ is not a spectral projection of $\A$, it still holds that it corresponds to a proposition ``$B\varepsilon\Ga$'' about the value of some physical quantity $B$ represented by a self-adjoint operator $\hat B$ in $V$, and since $\P<\dastoo{V}{P}$, this proposition can still be seen as a coarse-graining of the original proposition ``$\Ain\De$'' corresponding to $\P$ (even though $A\neq B$).

Every projection is contained in some abelian subalgebra, so there is always some $V$ such that $\dastoo{V}{P}=\P$. From the perspective of this context, no coarse-graining takes place.

We call the original proposition ``$\Ain\De$'' that we want to represent the \emph{global} proposition, while a proposition ``$B\varepsilon\Ga$'' corresponding to the projection $\dastoo{V}{P}$ is called a \emph{local} proposition. `Local' here again means `at the abelian part $V$' and has no spatio-temporal connotation.

We can collect all the local approximations into a mapping
\begin{equation}
			\P\mapsto(\dastoo{V}{P})_{V\in\VN}
\end{equation}
whose component at $V$ is given by \eq{Eq_dastooVP}. From every global proposition we thus obtain one coarse-grained local proposition for each context $V$. In the next step, we will consider the whole collection of coarse-grained local propositions, relate it to a suitable subobject of the spectral presheaf, and regard this object as the representative of the global proposition.

\paragraph{Daseinisation of projections.} Let $V$ be a context. Since $V$ is an abelian von Neumann algebra, the lattice $\PV$ of projections in $V$ is a distributive lattice. Moreover, $\PV$ is complete and orthocomplemented.  We saw in \eq{alpha} that there is a lattice isomorphism between $\PV$ and $\mc{C}l(\Sig_V)$, the lattice of clopen subsets of $\Sig_V$.

We already have constructed a family $(\dastoo{V}{P})_{V\in\VN}$ of projections from a single projection $\P\in\PN$ representing a global proposition. For each context $V$, we now consider the clopen subset $\alpha(\dastoo{V}{P})=S_{\dastoo{V}{P}}\subseteq\Si_V$. We thus obtain a family $(S_{\dastoo{V}{P}})_{V\in\VN}$ of clopen subsets, one for each context. One can show that whenever $V'\subset V$, then
\begin{equation}
			S_{\dastoo{V}{P}}|_{V'}=\{\ld|_{V'} \mid \ld\in S_{\dastoo{V}{P}}\}=S_{\dastoo{V'}{P}},
\end{equation}
that is, the clopen subsets in the family $(S_{\dastoo{V}{P}})_{V\in\VN}$ `fit together' under the restriction mappings of the spectral presheaf $\Sig$ (see Thm. 3.1 in \cite{DI08}). This means that the family $(S_{\dastoo{V}{P}})_{V\in\VN}$ of clopen subsets, together with the restriction mappings between them, forms a \emph{subobject} of the spectral presheaf. This particular subobject will be denoted as $\ps{\das{P}}$ and is called the \emph{daseinisation of $\P$}.

A subobject of the spectral presheaf (or any other presheaf) is the analogue of a subset of a space. Concretely, for each $V\in\VN$, the set $S_{\dastoo{V}{P}}$ is a subset of the Gel'fand spectrum of $V$. Moreover, the subsets for different $V$ are not arbitrary, but fit together under the restriction mappings of the spectral presheaf.

The collection of all subobjects of the spectral presheaf is a complete \emph{Heyting algebra} and is denoted as $\Sub{\Sig}$. It is the analogue of the power set of the state space of a classical system.

Let $\ps S$ be a subobject of the spectral presheaf $\Sig$ such that the component $S_V\subseteq\Sig_V$ is a clopen subset for all $V\in\VN$. We call such a subobject a \emph{clopen} subobject. All subobjects obtained from daseinisation of projections are clopen. One can show that the clopen subobjects form a complete Heyting algebra $\Subcl{\Sig}$ (see Thm. 2.5 in \cite{DI08}). This complete Heyting algebra can be seen as the analogue of the $\sigma$-complete Boolean algebra of Borel subsets of the state space of a classical system. In our constructions, the Heyting algebra $\Subcl{\Sig}$ will be more important than the bigger Heyting algebra $\Sub{\Sig}$, just as in classical physics, where the Boolean algebra of measurable subsets of state space is technically more important than the algebra of all subsets.

Starting from a global proposition ``$\Ain\De$'' with corresponding projection $\P$, we have constructed the clopen subobject $\ps{\das{P}}$ of the spectral presheaf. This subobject, the analogue of the measurable subset $f_A^{-1}(\De)$ of the state space of a classical system, is the representative of the global proposition ``$\Ain\De$''. The following mapping is called \emph{daseinisation of projections}:
\begin{eqnarray}
		\ps\delta:\PN &\longrightarrow& \Subcl{\Sig}\\			\nonumber
		\P &\longmapsto& \ps{\das{P}}.
\end{eqnarray}
It has the following properties:
\begin{itemize}
	\item[(1)] If $\P<\Q$, then $\ps{\delta(\P)}<\ps{\delta(\Q)}$, i.e., daseinisation is order-preserving;
	\item[(2)] the mapping $\ps{\delta}:\PN\rightarrow\Subcl{\Sig}$ is injective, that is, two (inequivalent) propositions correspond to two different subobjects;
	\item[(3)] $\ps{\delta(\hat 0)}=\ps 0$, the empty subobject, and $\ps{\delta(\hat 1)}=\Sig$. The trivially false proposition is represented by the empty subobject, the trivially true proposition is represented by the whole of $\Sig$.
	\item[(4)] for all $\P,\Q\in\PN$, it holds that $\ps{\delta(\P\vee\Q)}=\ps{\delta(\P)}\vee\ps{\delta(\Q)}$, that is, daseinisation preserves the disjunction (Or) of propositions;
	\item[(5)] for all $\P,\Q\in\PN$, it holds that $\ps{\delta(\P\wedge\Q)}\leq\ps{\delta(\P)}\wedge\ps{\delta(\Q)}$, that is, daseinisation does not preserve the conjunction (And) of propositions;
	\item[(6)] in general, $\ps{\delta(\P)}\wedge\ps{\delta(\Q)}$ is not of the form $\ps{\delta(\hat R)}$ for a projection $\hat R\in\PN$, and daseinisation is not surjective.
\end{itemize}
The domain of the mapping $\ps\delta$, the lattice $\PN$ of projections in the von Neumann algebra $\N$ of physical quantities, is the lattice that Birkhoff and von Neumann suggested as the algebraic representative of propositional quantum logic (\cite{BvN36}; in fact, they considered the case $\N=\BH$.) Thus daseinisation of projections can be seen as a `translation' mapping between ordinary, Birkhoff-von Neumann quantum logic, which is based upon the non-distributive lattice of projections $\PN$, and the topos form of propositional quantum logic, which is based upon the distributive lattice $\Subcl{\Sig}$. The latter more precisely is a Heyting algebra.

The quantum logic aspects of the topos formalism and the physical and conceptual interpretation of the relations above are discussed in some detail in \cite{Doe09}, building upon \cite{IB98,DI(2),DI08}. The resulting new form of quantum logic has many attractive features and avoids a number of the well-known conceptual problems of standard quantum logic \cite{DCG02}.

\begin{example}
Let $\Hi=\bbC^3$, and let $S_z$ be the physical quantity `spin in $z$-direction'. We consider the proposition ``$S_z\varepsilon(-0.1,0.1)$'', i.e., ``the spin in $z$-direction has a value between $-0.1$ and $0.1$''. (Since the eigenvalues are $-\frac{1}{\sqrt 2},0$ and $\frac{1}{\sqrt 2}$, this amounts to saying that the spin in $z$-direction is $0$.) The self-adjoint operator representing $S_z$ is
\begin{equation}
			\hat S_z=\frac{1}{\sqrt{2}}
			\left(\begin{array}
						[c]{ccc}%
						1 & 0 & 0\\
						0 & 0 & 0\\
						0 & 0 & -1
			\end{array}\right),
\end{equation}
and the eigenvector corresponding to the eigenvalue $0$ is $\ket\psi=(0,1,0)$. The projection $\P:=\hat E[S_z\varepsilon(-0.1,0.1)]=\ket\psi\bra\psi$ onto this eigenvector hence is
\begin{equation}
			\P=
			\left(\begin{array}
						[c]{ccc}%
						0 & 0 & 0\\
						0 & 1 & 0\\
						0 & 0 & 0
			\end{array}\right).
\end{equation}
There is exactly one context $V\in\mc{V}(\bbC^3)$ that contains the operator $\hat S_z$, namely the (maximal) context $V_{\hat S_z}=\{\P_1,\P_2,\P_3\}''$ generated by the projections 
\begin{equation}
			\P_1=
			\left(\begin{array}
						[c]{ccc}%
						1 & 0 & 0\\
						0 & 0 & 0\\
						0 & 0 & 0
			\end{array}\right),\ \ \
			\P_2=\P=
			\left(\begin{array}
						[c]{ccc}%
						0 & 0 & 0\\
						0 & 1 & 0\\
						0 & 0 & 0
			\end{array}\right),\ \ \
			\P_3=
			\left(\begin{array}
						[c]{ccc}%
						0 & 0 & 0\\
						0 & 0 & 0\\
						0 & 0 & 1
			\end{array}\right).			
\end{equation}
In the following, we consider different kinds of contexts $V$ and the approximations $\dasout{\P}{V}$ of $\P=\P_2$ to these contexts.

\textbf{(1) $V_{\hat S_z}$ and its subcontexts:} Of course, $\dasout{\P}{V_{\hat S_z}}=\P$. Let $V_{\P_1}=\{\P_1,\P+\P_3\}''$, then $\dasout{\P}{V_{\P_1}}=\P+\P_3$. Analogously, $\dasout{\P}{V_{\P_3}}=\P+\P_3$, but $\dasout{\P}{V_{\P}}=\P$, since $V_\P$ contains the projection $\P$.

\textbf{(2) Maximal contexts which share the projection $\P$ with $V_{\hat S_z}$, and their subcontexts:} Let $\Q_1,\Q_3$ be two orthogonal rank-$1$ projections that are both orthogonal to $\P$ (e.g. projections obtained from $\P_1,\P_3$ by a rotation about the axis determined by $\P$). Then $V=\{\Q_1,\P,\Q_3\}''$ is a maximal algebra that contains $\P$, so $\dasout{\P}{V}=\P$. Clearly, there are uncountably many contexts $V$ of this form.

For projections $\Q_1,\P,\Q_3$ as described above, we have $\dasout{\P}{V_{\Q_1}}=\P+\Q_3$, $\dasout{\P}{V_{\Q_3}}=\Q_1+\P$ (and $\dasout{\P}{V_\P}=\P$, as mentioned before).

\textbf{(3) Contexts which contain a rank-$2$ projection $\Q$ such that $\Q>\P$:} We first note that there are uncountably many projections $\Q$ of this form.
\begin{enumerate}
	\item [(a)] The trivial cases of contexts containing $\Q>\P$ are (1) maximal contexts containing $\P$; and (2) the context $V_{\Q}=V_{\hat 1-\Q}$.
	\item [(b)] There are also maximal contexts which do not contain $\P$, but do contain $\Q$: let $\Q_1,\Q_2$ be two orthogonal rank-$1$ projections such that $\Q_1+\Q_2=\Q$, where $\Q_1,\Q_2\neq\P$. Let $\Q_3:=\hat 1-(\Q_1+\Q_2)$. For any given $\Q$, there are uncountably many contexts of the form $V=\{\Q_1,\Q_2,\Q_3\}$ (since one can rotate around the axis determined by $\Q_3$).
\end{enumerate}
For all these contexts, $\dasout{\P}{V}=\Q$, of course.

\textbf{(4) Other contexts:} All other contexts neither contain $\P$ nor a projection $\Q>\P$ that is not the identity $\hat 1$. This means that the smallest projection larger than $\P$ contained in such a context is the identity, so $\dasout{\P}{V}=1$.

We have constructed $\dasout{\P}{V}$ for all contexts. Recalling that for each context $V$, there is a lattice isomorphism $\alpha$ from the projections to the clopen subsets of the Gel'fand spectrum $\Si_V$ (see \eq{alpha}), we can easily write down the clopen subsets $S_{\dasout{\P}{V}}\;(V\in\mc V (\bbC^3))$ corresponding to the projections $\dasout{\P}{V}\;(V\in\mc V (\bbC^3))$. The daseinisation $\ps{\das{P}}$ of the projection $\P$ is nothing but the collection $(S_{\dasout{\P}{V}})_{V\in\mc V (\bbC^3)}$ (together with the canonical restriction mappings described in \eq{Restriction} and Example \ref{Ex_SigOfB(C3)}).
\end{example}

\section{Representation of physical quantities -- Daseinisation of self-adjoint operators}	\label{S5}	\label{Sec_RepOfPhysQuants}
\subsection{Physical quantities as arrows}			\label{SubS_PhysQuantitiesAsArrows}
In classical physics, a physical quantity $A$ is represented by a function $f_A$ from the state space $\mc{S}$ of the system to the real numbers. Both the state space and the real numbers are sets (with extra structure), so they are objects in the topos $\Set$ of sets and functions.\footnote{We will not give the technical definition of a topos here. The idea is that a topos is a category that is structurally similar to $\Set$, the category of small sets and functions. Of course, $\Set$ itself is a topos.} The function $f_A$ representing a physical quantity is an arrow in the topos $\Set$.

The topos reformulation of quantum physics uses structures in the topos $\SetVNop$, the topos of presheaves over the context category $\VN$. The objects in the topos, called presheaves, are the analogues of sets, and the arrows between them, called natural transformations, are the analogues of functions. We will now represent a physical quantity $A$ by a suitable arrow, denoted $\dasB{\A}$, from the spectral presheaf $\Sig$ to some other presheaf related to the real numbers. Indeed, there is a real-number object $\ps{\bbR}$ in $\SetVNop$ that is very much like the familiar real numbers: for every $V$, the component $\ps{\bbR}_V$ simply is the ordinary real numbers, and all restriction functions are just the identity. But the Kochen-Specker theorem tells us that it is impossible to assign (sharp) real values to all physical quantities, which suggests that the presheaf $\ps{\bbR}$ is \emph{not} the right codomain for our arrows $\dasB{\A}$ representing physical quantities.

Instead, we will be using a presheaf, denoted as $\Rlr$, that takes into account that (a) one cannot assign sharp values to all physical quantities in any given state in quantum theory; and (b) we have coarse-graining. We first explain how coarse-graining will show up: if we want to assign a `value' to a physical quantity $A$, then we again have to find some global expression, involving all contexts $V\in\VN$. Let $\A\in\N_\sa$ be the self-adjoint operator representing $A$. For those contexts $V$ that contain $\A$, there is no problem, but if $\A\notin V$, we will have to approximate $\A$ by a self-adjoint operator \emph{in} $V$. Actually, we will use two approximations: one self-adjoint operator in $V$ approximating $\A$ from below (in a suitable order, specified below), and one operator approximating $\A$ from above. Daseinisation of self-adjoint operators is nothing but the technical device achieving this approximation.

We will show that the presheaf $\Rlr$ basically consists of real intervals, to be interpreted as `unsharp values'. Upon coarse-graining, the intervals can only get bigger. Assume we have given a physical quantity $A$, some state of the system and two contexts $V,V'$ such that $V'\subset V$. We use the state to assign some interval $[a,b]$ as an `unsharp value' to $A$ at $V$. Then, at $V'$, we will have a bigger interval $[c,d]\supseteq[a,b]$, being an even more unsharp value of $A$. It it important to note, though, that every self-adjoint operator is contained in some context $V$. If the state is an eigenstate of $\A$ with eigenvalue $a$, then at $V$, we will assign the `interval' $[a,a]$ to $A$. In this sense, the eigenvector-eigenvalue link is preserved, regardless of the unsharpness in values introduced by coarse-graining. As always in the topos scheme, it is of central importance that the interpretationally relevant structures -- here, the `values' of physical quantities -- are global in nature. One has to consider all contexts at once, and not just argue locally at some context $V$, which will necessarily only give partial information.

\paragraph{Approximation in the spectral order.} Let $\A\in\N_\sa$ be a self-adjoint operator, and let $V\in\VN$ be a context. We assume that $\A\notin V_\sa$, so we have to approximate $\A$ in $V$. The idea is very simple: we take a pair of self-adjoint operators, consisting of the largest operator in $V_\sa$ smaller than $\A$ and the smallest operator in $V_\sa$ larger than $\A$. The only non-trivial point is the question which order on self-adjoint operators to use.

The most commonly used order is the \emph{linear order}: $\A\leq\hat B$ if and only if $\bra\psi\A\ket\psi\leq\bra\psi\hat B\ket\psi$ for all vector states $\bra\psi\_\ket\psi:\N\rightarrow\bbC$. Yet, we will approximate in another order on the self-adjoint operators, the so-called \emph{spectral order} \cite{Ols71,deG04}. This has the advantage that the spectra of the approximated operators are subsets of the spectrum of the original operator, which is not the case for the linear order in general.

The spectral order is defined in the following way: let $\A,\hat B\in\N_\sa$ be two self-adjoint operators in a von Neumann algebra $\N$, and let $(\hat E^\A_r)_{r\in\bbR},(\hat E^{\hat B}_r)_{r\in\bbR}$ be their spectral families (see e.g. \cite{KR83a}). Then
\begin{equation}			\label{Def_SpecOrder}
			\A\leq_s\hat B\text{\ \ \ if and only if\ \ \ }\forall r\in\bbR:\hat E^\A_r\geq\hat E^{\hat B}_r.
\end{equation}
Equipped with the spectral order, $\N_\sa$ becomes a boundedly complete lattice. The spectral order is coarser than the linear order, i.e., $\A\leq_s\hat B$ implies $\A\leq\hat B$. The two orders coincide on projections and for commuting operators.

Consider the set $\{\hat B\in V_\sa \mid \hat B\leq_s\A\}$, i.e., those self-adjoint operators in $V$ that are spectrally smaller than $\A$. Since $V$ is a von Neumann algebra, its self-adjoint operators form a boundedly complete lattice under the spectral order, so the above set has a well-defined maximum with respect to the spectral order. It is denoted as
\begin{equation}		\label{dastoiVA}
			\dastoi{V}{A}:=\bigvee\{\hat B\in V_\sa \mid \hat B\leq_s\A\}.
\end{equation}
Similarly, we define
\begin{equation}		\label{dastooVA}
			\dastoo{V}{A}:=\bigwedge\{\hat B\in V_\sa \mid \hat B\geq_s\A\},
\end{equation}
which is the smallest self-adjoint operator in $V$ that is spectrally larger than $\A$.

Using the definition \eq{Def_SpecOrder} of the spectral order, it is easy to describe the spectral families of the operators $\dastoi{V}{A}$ and $\dastoo{V}{A}$:
\begin{equation}			\label{Eq_SpecFamOutDas}
			 \forall r\in\bbR:\hat E^{\dastoo{V}{A}}_r=\dasinn{\hat E^\A_r}{V}
\end{equation}
and
\begin{equation}
			\forall r\in\bbR:\hat E^{\dastoi{V}{A}}_r=\dasout{\hat E^\A_r}{V}.
\end{equation}
It turns out (see \cite{deG07,DI(2)}) that the spectral family defined by the former definition is right-continuous. The latter expression must be amended slightly if we want a right-continuous spectral family: we simply enforce right continuity by setting
\begin{equation}			\label{Eq_SpecFamInnDas}
			\forall r\in\bbR:\hat E^{\dastoi{V}{A}}_r=\bigwedge_{s>r}\dasout{\hat E^\A_s}{V}.
\end{equation}
These equations show the close mathematical link between the approximations in the spectral order of projections and self-adjoint operators.

The self-adjoint operators $\dastoi{V}{A},\dastoo{V}{A}$ are the best approximations to $\A$ in $V_\sa$ with respect to the spectral order. $\dastoi{V}{A}$ approximates $\A$ from below, $\dastoo{V}{A}$ from above. Since $V$ is an abelian $C^*$-algebra, we can consider the Gel'fand transforms of these operators to obtain a pair
\begin{equation}
			(\fu{\dastoi{V}{A}},\fu{\dastoo{V}{A}})
\end{equation}
of continuous functions from the Gel'fand spectrum $\Si_V$ of $V$ to the real numbers. One can show that \cite{deG05b,DI(3)}
\begin{equation}
			\spec{\dastoi{V}{A}}\subseteq\spec{\A},\ \ \ \ \spec{\dastoo{V}{A}}\subseteq\spec{\A},
\end{equation}
so the Gel'fand transforms actually map into the spectrum $\spec{\A}$ of the original operator $\A$. Let $\ld\in\Si_V$ be a pure state of the algebra $V$.\footnote{We remark that in general, this cannot be identified with a state of the non-abelian algebra $\N$, so we still have a local argument here.} Then we can evaluate $(\fu{\dastoi{V}{A}}(\ld),\fu{\dastoo{V}{A}}(\ld))$ to obtain a pair of real numbers in the spectrum of $\A$. The first number is smaller than the second, since the first operator $\dastoi{V}{A}$ approximates $\A$ from below in the spectral order (and $\dastoi{V}{A}\leq_s\A$ implies $\dastoi{V}{A}\leq\A$), while the second operator $\dastoo{V}{A}$ approximates $\A$ from above. The pair $(\fu{\dastoi{V}{A}}(\ld),\fu{\dastoo{V}{A}}(\ld))$ of real numbers is identified with the interval $[\fu{\dastoi{V}{A}}(\ld),\fu{\dastoo{V}{A}}(\ld)]$ and interpreted as the component at $V$ of the unsharp value of $\A$ in the (local) state $\ld$.

\paragraph{The presheaf of `values'.} We now see in which way intervals show up as `values'. If we want to define a presheaf encoding this, we encounter a certain difficulty. It is no problem to assign to each context $V$ the collection $\mathbb{IR}_V$ of real intervals, but if $V'\subset V$, then we need a restriction mapping from $\mathbb{IR}_V$ to $\mathbb{IR}_{V'}$. Since both sets are the same, the naive guess is to take an interval in $\mathbb{IR}_V$ and to map it to the same interval in $\mathbb{IR}_{V'}$. But we already know that the `values' of our physical quantities come from approximations of a self-adjoint operator $\A$ to all the contexts. If $V'\subset V$, then in general we have $\dastoi{V'}{A}<_s\dastoi{V'}{A}$,\footnote{This, of course, is nothing but coarse-graining on the level of self-adjoint operators.} which implies $\fu{\dastoi{V'}{A}}(\ld|_{V'})\leq\fu{\dastoi{V}{A}}(\ld)$ for $\ld\in\Si_V$. Similarly, $\fu{\dastoo{V'}{A}}(\ld|_{V'})\geq\fu{\dastoo{V}{A}}(\ld)$, so the intervals that we obtain get bigger when we go from larger contexts to smaller ones, due to coarse-graining. Our restriction mapping from the intervals at $V$ to the intervals at $V'$ must take this into account.

The intervals at different contexts clearly also depend on the state $\ld\in\Si_V$ that one picks. Moreover, we want a presheaf that is not tied to the operator $\A$, but can provide `values' for all physical quantities and their corresponding self-adjoint operators.

The idea is to define a presheaf such that at each context $V$, we have all intervals (including those of the form $[a,a]$), plus all possible restrictions to the same or larger intervals at smaller contexts $V'\subset V$. While this sounds daunting, there is a very simple way to encode this: let $\downarrow\!\!V:=\{V'\in\VN \mid V'\subseteq V\}$. This is a partially ordered set. We now consider a function $\mu:\downarrow\!\!V\rightarrow\bbR$ that \emph{preserves the order}, that is, $V'\subset V$ implies $\mu(V')\leq\mu(V)$. Analogously, let $\nu:\downarrow\!\!V\rightarrow\bbR$ denote an \emph{order-reversing} function, that is, $V'\subset V$ implies $\mu(V')\geq\mu(V)$. Additionally, we assume that for all $V'\subseteq V$, it holds that $\mu(V')\leq\nu(V')$. The pair $(\mu,\nu)$ thus gives one interval $(\mu(V'),\nu(V'))$ for each context $V'\subseteq V$, and `going down the line' from $V$ to smaller subalgebras $V',V'',...$, these intervals can only get bigger (or stay the same). Of course, each pair $(\mu,\nu)$ gives a specific sequence of intervals (for each given $\ld$). In order to have all possible intervals and sequences built from them, we simply consider the collection of all pairs of order-preserving and -reversing functions.

This is formalised in the following way: let $\Rlr$ be the presheaf given
\begin{itemize}
	\item [(a)] on objects: for all $V\in\VN$, let
	\begin{equation}
				\Rlr_V:=\{(\mu,\nu) \mid \mu,\nu:\downarrow\!\!V\rightarrow\bbR,\mu\text{ order-preserving, }\nu\text{ order-reversing},\ \mu\leq\nu\};
	\end{equation}
	\item [(b)] on arrows: for all inclusions $i_{V'V}:V'\rightarrow V$, let
	\begin{eqnarray}
				\Rlr(i_{V'V}):\Rlr_V &\longrightarrow& \Rlr_{V'}\\			\nonumber
				(\mu,\nu) &\longmapsto& (\mu|_{V'},\nu|_{V'}).
	\end{eqnarray}
\end{itemize}
The presheaf $\Rlr$ is where physical quantities take their `values' in the topos version of quantum theory. It is the analogue of the set of real numbers, where physical quantities in classical physics take their values.

\paragraph{Daseinisation of a self-adjoint operator.} We now start to assemble the approximations to the self-adjoint operator $\A$ into an arrow from the spectral presheaf $\Sig$ to the presheaf $\Rlr$ of values. Such an arrow is a so-called natural transformation (see e.g. \cite{McL98}), since these are the arrows in the topos $\SetVNop$. For each context $V\in\VN$, we define a function, which we will denote as $\dasB{\A}_V$, from the Gel'fand spectrum $\Sig_V$ of $V$ to the collection $\Rlr_V$ of pairs of order-preserving and -reversing functions from $\downarrow\!\!V$ to $\bbR$. In a second step, we will see that these functions fit together in the appropriate sense to form a natural transformation.

The function $\dasB{\A}_V$ is constructed as follows: let $\ld\in\Sig_V$. In each $V'\in\downarrow\!\!V$, we have one approximation $\delta^i(\A)_{V'}$ to $\A$. If $V''\subset V'$, then $\delta^i(\A)_{V''}\leq_s\delta^i(\A)_{V'}$, which implies $\delta^i(\A)_{V''}\leq\delta^i(\A)_{V'}$. One can evaluate $\ld$ at $\delta^i(\A)_{V'}$ for each $V'$, giving an order-preserving function
\begin{align}			\label{mu_ld}
				\mu_\ld:\downarrow\!\!V &\longrightarrow \bbR\\			 \nonumber
				V' &\longmapsto \ld|_{V'}(\delta^i(\A)_{V'})=\ld(\delta^i(\A)_{V'}).
\end{align}
Similarly, using the approximations $\delta^o(\A)_{V'}$ to $\A$ from above, we obtain an order-reversing function
\begin{align}			\label{nu_ld}
				\nu_\ld:\downarrow\!\!V &\longrightarrow \bbR\\			\nonumber
				V' &\longmapsto \ld|_{V'}(\delta^o(\A)_{V'})=\ld(\delta^o(\A)_{V'}).
\end{align}
This allows us to define the desired function:
\begin{align}			\label{dasB(A)_V}
			\dasB{\A}_V:\Sig_V &\longrightarrow \Rlr_V\\			\nonumber
			\ld &\longmapsto (\mu_\ld,\nu_\ld).
\end{align}

Obviously, we obtain one function $\dasB{\A}_V$ for each context $V\in\VN$. The condition for these functions to form a natural transformation is the following: whenever $V'\subset V$, one must have $\dasB{\A}_V(\ld)|_{V'}=\dasB{\A}_{V'}(\ld|_{V'})$. In fact, this equality follows trivially from the definition. We thus have arrived at an arrow
\begin{equation}
			\dasB{\A}:\Sig\longrightarrow\Rlr
\end{equation}
representing a physical quantity $A$. The arrow $\dasB{\A}$ is called the \emph{daseinisation of $\A$}. This arrow is the analogue of the function $f_A:\mc{S}\rightarrow\bbR$ from the state space to the real numbers representing the physical quantity $A$ in classical physics.

\paragraph{Pure states and assignment of `values' to physical quantities.} Up to now, we have always argued with elements $\ld\in\Sig_V$. This is a local argument, and we clearly need a global representation of states. Let $\psi$ be a unit vector in Hilbert space, identified with the pure state $\bra\psi\_\ket\psi:\N\rightarrow\bbC$ as usual. In the topos approach, such a pure state is represented by the subobject $\ps{\delta(\P_\psi)}$ of the spectral presheaf $\Sig$, i.e., the daseinisation of the projection $\P_\psi$ onto the line $\bbC\psi$ in Hilbert space. Details can be found in \cite{DI08}.

The subobject $\ps\wpsi:=\ps{\delta(\P_\psi)}$ is called the \emph{pseudo-state} associated to $\psi$. It is the analogue of a point $s\in\mc{S}$ of the state space of a classical system. Importantly, the pseudo-state is not a global element of the presheaf $\Sig$. Such a global element would be the closest analogue of a point in a set, but as Isham and Butterfield observed \cite{IB98}, the spectral presheaf $\Sig$ has no global elements at all -- a fact that is equivalent to the Kochen-Specker theorem. Nonetheless, subobjects of $\Sig$ of the form $\ps\wpsi=\ps{\delta(\P_\psi)}$ are minimal subobjects in an appropriate sense and as such are as close to points as possible (see \cite{DI08,Doe08}).

The `value' of a physical quantity $A$, represented by an arrow $\dasB{\A}$, in a state described by $\ps\wpsi$ then is
\begin{equation}
            \dasB{\A}(\ps\wpsi).
\end{equation}
This of course is the analogue of the expression $f_A(s)$ in classical physics, where $f_A$ is the real-valued function on the state space $\mc{S}$ of the system representing the physical quantity $A$, and $s\in\mc{S}$ is the state of the system. 

The expression $\dasB{\A}(\ps\wpsi)$ turns out to describe a considerably more complicated object than the classical expression $f_A(s)$. One can easily show that $\dasB{\A}(\ps\wpsi)$ is a subobject of the presheaf $\Rlr$ of `values', see section 8.5 in \cite{DI08}. Concretely, it is given at each context $V\in\VN$ by
\begin{equation}
            (\dasB{\A}(\ps\wpsi))_V=\dasB{\A}_V(\ps\wpsi_V)=\{\dasB{\A}_V(\ld) \mid \ld\in\ps\wpsi_V\}.
\end{equation}
As we saw above, each $\dasB{\A}_V(\ld)$ is a pair of functions $(\mu_\ld,\nu_\ld)$ from the set $\downarrow\!\!V$ to the reals, and more specifically to the spectrum $\spec{\A}$ of the operator $\A$. The function $\mu_\ld$ is order-preserving, so it takes smaller and smaller values as one goes down from $V$ to smaller subalgebras $V',V'',...$, while $\nu_\ld$ is order-reversing and takes larger and larger values. Together, they determine one real interval $[\mu_\ld(V'),\nu_\ld(V')]$ for every context $V'\in\downarrow\!\!V$. Every $\ld\in\ps\wpsi_V$ determines one such sequence of nested intervals. The `value' of the physical quantity $A$ in the given state is the collection of all these sequences of nested intervals for all the different contexts $V\in\VN$.

We conjecture that the usual expectation value $\bra\psi\A\ket\psi$ of the physical quantity $A$, which of course is a single real number, can be calculated from the `value' $\dasB{\A}(\ps\wpsi)$. It is easy to see that each real interval showing up in (the local expressions of) $\dasB{\A}(\ps\wpsi)$ contains the expectation value $\bra\psi\A\ket\psi$.

In \cite{Doe08}, it was shown that arbitrary quantum states (not just pure ones) can be represented by probability measures on the spectral presheaf. It was also shown how the expectation values of physical quantities can be calculated using these measures.

Instead of closing this section by showing how the relevant structures look like concretely for our standard example $\N=\mc{B}(\bbC^3)$, we will introduce an efficient way of calculating the approximations $\dastoi{V}{A}$ and $\dastoo{V}{A}$ first, which works uniformly for all contexts $V\in\VN$. This will simplify the calculation of the arrow $\dasB{\A}$ significantly.

\subsection{Antonymous and observable functions}
Our aim is to define two functions $g_\A,f_\A$ from which we can calculate the approximations $\dastoi{V}{A}$ resp. $\dastoo{V}{A}$ to $\A$ for all contexts $V\in\VN$. More precisely, we will describe the Gel'fand transforms of the operators $\dastoi{V}{A}$ and $\dastoo{V}{A}$ (for all $V$). These are functions from the Gel'fand spectrum $\Si_V$ to the reals. The main idea, due to de Groote, is to use the fact that each $\ld\in\Si_V$ corresponds to a maximal \emph{filter} in the projection lattice $\PV$ of $V$. Hence, we will consider the Gel'fand transforms as functions on filters, and by `extending' each filter in $\PV$ to a filter in $\PN$ (non-maximal in general), we will arrive at functions on filters in $\PN$, the projection lattice of the non-abelian von Neumann algebra $\N$. First, we collect a few basic definitions and facts about filters:

\paragraph{Filters.} Let $\mathbb{L}$ be a lattice with zero element $0$. A subset $F$ of elements of $\mathbb{L}$ is a \emph{(proper) filter} (or \emph{(proper) dual ideal}) if (i) $0\notin F$, (ii) $a,b\in F$ implies $a\wedge b\in F$ and (iii) $a\in F$ and $b\geq a$ imply $b\in F$. When we speak of filters, we always mean proper filters. We write $\mc{F}(\mathbb{L})$ for the set of filters in a lattice $\mathbb{L}$. If $\N$ is a von Neumann algebra, then we write $\mc{F}(\N)$ for $\mc{F}(\PN)$.

A subset $B$ of elements of $\mathbb{L}$ is a \emph{filter base} if (i) $0\neq B$ and (ii) for all $a,b\in B$, there exists a $c\in B$ such that $c\leq a\wedge b$.

Maximal filters and maximal filter bases in a lattice $\mathbb{L}$ with $0$ exist by Zorn's lemma, and each maximal filter is a maximal filter base and vice versa. The maximal filters in the projection lattice of a von Neumann algebra $\N$ are denoted by $\mc{Q}(\N)$. Following de Groote \cite{deG05c}, one can equip this set with a topology and consider its relations to well-known constructions (in particular, Gel'fand spectra of abelian von  Neumann algebras). Here, we will not consider the topological aspects.

If $V$ is an abelian subalgebra of a von Neumann algebra $\N$, then the projections in $V$ form a distributive lattice $\PV$, that is, for all $\P,\Q,\hat{R}\in\PV$ we have
\begin{eqnarray}
			\P\wedge(\Q\vee\hat{R})  &=& (\P\wedge\Q)\vee(\P\wedge\hat{R}),\\
			\P\vee(\Q\wedge\hat{R})  &=& (\P\vee\hat{Q})\wedge(\P\vee\hat{R}).
\end{eqnarray}
The projection lattice of a von Neumann algebra is distributive if and only if the algebra is abelian.

A maximal filter in a complemented, distributive lattice $\mathbb{L}$ is called an \emph{ultrafilter}, hence $\mc{Q}(\mathbb{L})$ is the space of ultrafilters in $\mathbb{L}$. We will not use the notion `ultrafilter' for maximal filters in \emph{non}-distributive lattices like the projection lattices of non-abelian von Neumann algebras.

The characterising property of an ultrafilter $U$ in a distributive, complemented lattice $\mathbb{L}$ is that for each element $a\in\mathbb{L}$, either $a\in U$ or $a^{c}\in U$, where $a^{c}$ denotes the complement of $a$. This can easily be seen: a complemented lattice has a maximal element $1$, and we have $a\vee a^{c}=1$ by definition. Let us assume that $U$ is an ultrafilter, but $a\notin U$ and $a^{c}\notin U$. In particular, $a\notin U$ means that there is some $b\in U$ such that $b\wedge a=0$. Using distributivity of the lattice $\mathbb{L}$, we get
\begin{equation}
			b=b\wedge(a\vee a^{c})=(b\wedge a)\vee(b\wedge a^{c})=b\wedge a^{c},
\end{equation}
so $b\leq a^{c}$. Since $b\in U$ and $U$ is a filter, this implies $a^{c}\in U$, contradicting our assumption. If $\mathbb{L}$ is the projection lattice $\PV$ of an abelian von Neumann algebra $V$, then the maximal element is the identity operator $\hat{1}$ and the complement of a projection is given as $\P^{c}:=\hat{1}-\P$. Each ultrafilter $q\in\mc{Q}(V)$ hence contains either $\P$ or $\hat{1}-\P$ for all $\P\in\PV$.

\paragraph{Elements of the Gel'fand spectrum and filters.} We already saw that for all projections $\P\in\PV$ and all elements $\ld$ of the Gel'fand spectrum $\Si_V$ of $V$, it holds that $\ld(\P)\in\{0,1\}$. Given $\ld\in\Si_V$, it is easy to construct a maximal filter $F_\ld$ in the projection lattice $\PV$ of $V$: let
\begin{equation}
			F_\ld:=\{\P\in\PV\mid\ld(\P)=1\}.
\end{equation}
It is clear that $\P\in F_\ld$ implies $\Q\in F_\ld$ for all projections $\Q\geq\P$. Let $\P_{1},\P_{2}\in F_{\ld}$. From the fact that $\ld$ is a multiplicative state of $V$, we obtain $\ld(\P_{1}\wedge\P_{2})=\ld(\P_{1}\P_{2})=\ld(\P_{1})\ld(\P_{2})=1$, so $\P_{1}\wedge\P_{2}\in F_{\ld}$. Finally, we have $1=\ld(\hat 1)=\ld(\P+\hat 1-\P)=\ld(\P)+\ld(\hat 1-\P)$ for all $\P\in\PV$, so either $\P\in F_\ld$ or $\hat 1-\P\in F_\ld$ for all $\P$, which shows that $F_{\ld}$ is a maximal filter in $\PV$, indeed. It is straightforward to see that the mapping
\begin{align}			\label{Eq_MaxFilterFromld}
			\beta:\Si_V &\longrightarrow \mathcal{Q}(V)\\			\nonumber
			\ld &\longmapsto F_\ld
\end{align}
from the Gel'fand spectrum of $V$ to the set $\mathcal{Q}(V)$ of maximal filters in $\mathcal{P}(V)$ is injective. This is the first step in the construction of a homeomorphism between these two spaces. De Groote has shown the existence of such a homeomorphism in Thm. 3.2 of \cite{deG05c}.

Given a filter $F$ in the projection lattice $\PV$ of some context $V$, one can `extend' it to a filter in the projection lattice $\PN$ of the full non-abelian algebra $\N$ of physical quantities. One simply defines the \textit{$\N$-cone over $F$} as
\begin{equation}
			\mc C_{\N}(F):=\uparrow\!\!F=\{\Q\in\PN \mid \exists\P\in F:\P\leq\Q\}.
\end{equation}
This is the smallest filter in $\PN$ containing $F$.

We need the following technical, but elementary lemma:
\begin{lemma}			\label{L_deltaT^-1(D)=Cone(D)}
Let $\N$ be a von Neumann algebra, and let $\mc{T}$ be a von Neumann subalgebra such that $\hat 1_{\mc{T}}=\hat 1_\N$. Let 
\begin{align}
			\delta_{\mc{T}}^{i}:\PN &\longrightarrow \mc{P(T)}\\			\nonumber
			\P &\longmapsto \dastoi{\mc{T}}{P}:=\bigvee\{\Q\in\mc{P(T)} \mid \Q\leq\P\}.
\end{align}			
Then, for all filters $F\in\mc{F}(\mc{T})$,
\begin{equation}
			(\delta_{\mc{T}}^{i})^{-1}(F)=\mc{C}_{\N}(F).
\end{equation}
\end{lemma}

\begin{proof}
If $\Q\in F\subset\mc{P(T)}$, then $(\delta_{\mc{T}}^{i})^{-1}(\Q)=\{\P\in\PN\mid\delta^{i}(\hat{P})_{\mc{T}}=\Q\}$. Let $\P\in\PN$ be such that there is a $\Q\in F$ with $\Q\leq\P$, i.e., $\P\in\mc{C}_{\N}(F)$. Then $\delta^{i}(\hat{P})_{\mc{T}}\geq\Q$, which implies $\delta^{i}(\hat{P})_{\mc{T}}\in F$, since $F$ is a filter in $\mc{P(T)}$. This shows that $\mc{C}_{\N}(F)\subseteq(\delta_{\mc{T}}^{i})^{-1}(F)$. Now let $\P\in\PN$ be such that there is no $\Q\in F$ with $\Q\leq\P$. Since $\delta^{i}(\P)_{\mc{T}}\leq\P$, there also is no $\hat{Q}\in F$ with $\Q\leq\delta^{i}(\P)_{\mc{T}}$, so $\P\notin(\delta_{\mc{T}}^{i})^{-1}(F)$. This shows that $(\delta_{\mc{T}}^{i})^{-1}(F)\subseteq\mc{C}_{\N}(F)$.
\end{proof}

\paragraph{Definition of antonymous and observable functions.}
Let $\N$ be a von Neumann algebra, and let $\A$ be a self-adjoint operator in $\A$. As before, we denote the spectral family of $\A$ as $(\hat E^\A_r)_{r\in\bbR}=\hat E^\A$.
\begin{definition}
The \emph{antonymous function of $\A$} is defined as
\begin{align}
			g_\A:\mc F (\N) &\longrightarrow \spec{\A}\\
			F &\longmapsto \on{sup}\{r\in\bbR \mid \hat 1-\hat E^\A_r\in F\}.
\end{align}
The \emph{observable function of $\A$} is
\begin{align}
			f_\A:\mc F (\N) &\longrightarrow \spec{\A}\\
			F &\longmapsto \on{inf}\{r\in\bbR \mid \hat E^\A_r\in F\}.
\end{align}
\end{definition}
It is straightforward to see that the range of these functions indeed is the spectrum $\spec{\A}$ of $\A$.

We observe that the mapping described in \eq{dastoiVA} can be generalised: when approximating a self-adjoint operator $\A$ with respect to the spectral order in a von Neumann subalgebra of $\N$, then one can use an arbitrary, not necessarily abelian subalgebra $\mc T$ of the non-abelian algebra $\N$. The only condition is that the unit elements in $\N$ and $\mc T$ coincide. We define
\begin{align}
			\delta^i_{\mc T}:\N_\sa &\longrightarrow \mc T_\sa\\		\nonumber
			\A &\longmapsto \dastoi{\mc T}{A}=\bigvee\{\hat B\in\mc T_\sa \mid \hat B\leq_s\A\}.
\end{align}
Analogously, generalising \eq{dastooVA}
\begin{align}
			\delta^o_{\mc T}:\N_\sa &\longrightarrow \mc T_\sa\\		\nonumber
			\A &\longmapsto \dastoo{\mc T}{A}=\bigwedge\{\hat B\in\mc T_\sa \mid \hat B\geq_s\A\}.
\end{align}
The following proposition is central:
\begin{proposition}			\label{P_g_AEncodesAllg_delta^i(A)}
Let $\A\in\N_{sa}$. For all von Neumann subalgebras $\mc{T}\subseteq\N$ such that $\hat 1_{\mc{T}}=\hat 1_\N$ and all filters $F\in\mc{F(T)}$, we have
\begin{equation}
			g_{\dastoi{\mc{T}}{A}}(F)=g_{\A}(\mc{C}_{\N}(F)).
\end{equation}
\end{proposition}

\begin{proof}
We have%
\begin{eqnarray*}
			g_{\delta^{i}(\A)_{\mc{T}}}(F) &=&  \sup\{r\in\bbR \mid \hat{1}-\hat{E}_{r}^{\dastoi{\mc{T}}{A}}\in F\}\\
			&\stackrel{\eq{Eq_SpecFamInnDas}}{=}&  \sup\{r\in\bbR \mid \hat{1}-\bigwedge_{\mu>r}\delta^{o}(\hat{E}_{\mu}^{\hat A})_{\mc{T}}\in F\}\\
			&=&  \sup\{r\in\bbR \mid \hat{1}-\delta^{o}(\hat{E}_{r}^{\hat A})_{\mc{T}}\in F\}\\
			&=&  \sup\{r\in\bbR \mid \delta^{i}(\hat{1}-\hat{E}_{r}^{\hat A})_{\mc{T}}\in F\}\\
			&=&  \sup\{r\in\bbR \mid \hat{1}-\hat{E}_{r}^{\hat A}\in(\delta_{\mc{T}}^{i})^{-1}(F)\}\\
			&\stackrel{\text{Lemma \ref{L_deltaT^-1(D)=Cone(D)}}}{=}&  \sup\{r\in\bbR \mid \hat{1}-\hat{E}_{r}^{\hat A}\in\mc{C}_{\N}(F)\}\\
			&=&  g_{\A}(\mc{C}_{\N}(F))
\end{eqnarray*}
\end{proof}

This shows that the antonymous function $g_{\A}:\mc{F}(\N)\rightarrow\spec{\A}$ of $\A$ encodes in a simple way \emph{all} the antonymous functions $g_{\delta^{i}(\A)_{\mc{T}}}:\mc{F(T)}\rightarrow\spec{\delta^{i}(\A)}_{\mc{T}}$ corresponding to the approximations $\delta^{i}(\A)_{\mc{T}}$ (from below in the spectral order) to $\A$ to von Neumann subalgebras $\mc{T}$ of $\N$. If, in particular, $\mc{T}$
is an abelian subalgebra, then the set $\mc{Q(T)}$ of \emph{maximal} filters in $\mc{D(T)}$ can be identified with the Gel'fand spectrum of $\mc{T}$ (see \cite{deG05b}), and $g_{\delta^{i}(\A)_{\mc{T}}}|_{\mc{Q(T)}}$ can be identified with the Gel'fand transform of
$\delta^{i}(\A)_{\mc{T}}$. Thus we get

\begin{corollary}
The antonymous function $g_{\A}:\mc{F}(\N)\rightarrow\spec{\A}$ encodes all the Gel'fand transforms of the inner daseinised operators of the form $\dastoi{V}{A}$, where $V\in\VN$ is an abelian von Neumann subalgebra of $\N$ (such that $\hat 1_V=\hat 1_\N$). Concretely, if $\mc{Q}(V)$ is the space of maximal filters in $\PV$ and $\Si_V$ is the Gel'fand spectrum of $V$, then we have an identification $\beta:\Sig_V\rightarrow\mc{Q}(V)$, $\ld\mapsto F_\ld$ (see \eq{Eq_MaxFilterFromld}). Let $\fu{\dastoi{V}{A}}:\Si_V\rightarrow\spec{\dastoi{V}{A}}$ be the Gel'fand transform of $\dastoi{V}{A}$. We can identify $\fu{\dastoi{V}{A}}$ with $g_{\dastoi{V}{A}}|_{\mc{Q}(V)}$, and hence, from Prop. \ref{P_g_AEncodesAllg_delta^i(A)},
\begin{equation}			\label{EqInnDasFromAnton}
			\fu{\dastoi{V}{A}}(\ld)=g_{\dastoi{V}{A}}(F_\ld)=g_{\A}(\mc{C}_{\N}(F_\ld)).
\end{equation}
\end{corollary}

Not surprisingly, there is a similar result for observable functions and outer daseinisation. This result was first proved in a similar form by de Groote in \cite{deG07}.

\begin{proposition}			\label{P_f_AEncodesAllf_delta^o(A)}
Let $\A\in\N_{sa}$. For all von Neumann subalgebras $\mc{T}\subseteq\N$ such that $\hat 1_{\mc{T}}=\hat 1_\N$ and all
filters $F\in\mc{F(T)}$, we have
\begin{equation}
			f_{\dastoo{\mc{T}}{A}}(F)=f_{\A}(\mc{C}_{\N}(F)).
\end{equation}
\end{proposition}

\begin{proof}
We have
\begin{eqnarray}
			f_{\dastoo{\mc{T}}{A}}(F) &=& \inf\{r\in\bbR \mid \hat{E}_{r}^{\dastoo{\mc{T}}{A}}\in F\}\\
			&\stackrel{\eq{Eq_SpecFamOutDas}}{=}& \inf\{r\in\bbR \mid \delta^{i}(\hat{E}_{r}^{\hat A})_{\mc{T}}\in F\}\\
			&=& \inf\{r\in\bbR \mid \hat{E}_{r}^{\hat A}\in(\delta_{\mc{T}}^{i})^{-1}(F)\}\\
			&\stackrel{\text{Lemma \ref{L_deltaT^-1(D)=Cone(D)}}}{=}& \inf\{r\in\bbR \mid \hat{E}_{r}^{\hat A}\in\mc{C}_{\N}(F)\}\\
			&=& f_{\A}(\mc{C}_{\N}(F)).
\end{eqnarray}
\end{proof}

This shows that the observable function $f_\A:\mc{F}(\N)\rightarrow\spec{\A}$ of $\A$ encodes in a simple way \emph{all} the observable functions $f_{\dastoo{\mc{T}}{A}}:\mc{F}(\N)\rightarrow\spec{\dastoo{\mc{T}}{A}}$ corresponding to approximations $\dastoo{\mc{T}}{A}$ (from above in the spectral order) to $\A$ to von Neumann subalgebras $\mc{T}$ of $\N$. If, in particular, $\mc{T}=V$ is an abelian subalgebra, then we get
\begin{corollary}
The observable function $f_{\A}:\mc{F}(\N)\rightarrow\spec{\A}$ encodes all the Gel'fand transforms of the outer daseinised operators of the form $\dastoo{V}{A}$, where $V\in\VN$ is an abelian von Neumann subalgebra of $\N$ (such that $\hat 1_V=\hat 1_\N$). Using the identification $\beta:\Sig_V\rightarrow\mc{Q}(V)$, $\ld\mapsto F_{\ld}$ between the Gel'fand and the space of maximal filters (see \eq{Eq_MaxFilterFromld}), we can identify the Gel'fand transform $\fu{\dastoo{V}{A}:\Sig_V\rightarrow
\spec{\dastoo{V}{A}}}$ of $\dastoo{V}{A}$ with $f_{\dastoo{V}{A}}|_{\mc{Q(T)}}$, and hence, from Prop. \ref{P_f_AEncodesAllf_delta^o(A)},
\begin{equation}
			\fu{\dastoo{V}{A}}(\ld)=f_{\dastoo{V}{A}}(F_{\ld})=f_\A(\mc{C}_{\N}(F_{\ld})).
\end{equation}
\end{corollary}

\paragraph{Vector states from elements of Gel'fand spectra.} As emphasised in section \ref{SubS_PhysQuantitiesAsArrows}, an element of $\ld$ of the Gel'fand spectrum $\Sig_V$ of some context $V$ is only a local state: it is a pure state of the abelian subalgebra $V$, but it is not a state of the full non-abelian algebra $\N$ of physical quantities. We will now show how certain local states $\ld$ can be extended to global states, i.e., states of $\N$.

Let $V$ be a context that contains at least one rank-$1$ projection.\footnote{Depending on whether the von Neumann algebra $\N$ contains rank-$1$ projections, there may or may not be such contexts. For the case $\N=\BH$, there of course are many rank-$1$ projections and hence contexts $V$ of the required form.} Let $\P$ be such a projection, and let $\psi\in\Hi$ be a unit vector such that $\P$ is the projection onto the ray $\bbC\psi$. ($\psi$ is fixed up to a phase.) We write $\P=\P_\psi$. Then
\begin{equation}			\label{F_P_psi}
			F_{\P_\psi}:=\{\Q\in\PV \mid \Q\geq\P_\psi\}
\end{equation}
clearly is an ultrafilter in $\PV$. It thus corresponds to some element $\ld_\psi$ of the Gel'fand spectrum $\Sig_V$ of $V$. Being an element of the Gel'fand spectrum, $\ld_\psi$ is a pure state of $V$. The cone
\begin{equation}
			\mc C_{\N}(F_{\P_\psi})=\{\Q\in\PN \mid \Q\geq\P_\psi\}
\end{equation}
is a maximal filter in $\PN$. It contains all projections in $\N$ that represent propositions ``$\Ain\De$'' that are (totally) true in the vector state $\psi$ on $\N$.\footnote{Here, we again use the usual identification between unit vectors and vector states, $\psi\mapsto\bra\psi\_\ket\psi$.} This vector state is uniquely determined by the cone $\mc C_{\N}(F_{\P_\psi})$. The `local state' $\ld_\psi$ on the context $V\subset\N$ can thus be extended to a `global' vector state. On the level of filters, this extension corresponds to the cone construction.

For a finite-dimensional Hilbert space $\Hi$, one necessarily has $\N=\BH$. Let $V$ be an arbitrary maximal context, and let $\ld\in\Sig_V$. Then $\ld$ is of the form $\ld=\ld_{\P_\psi}$ for some unit vector $\psi\in\Hi$ and corresponding rank-$1$ projection $\P_\psi$. Hence, every such $\ld_{\P_\psi}$ can be extended to a vector state $\psi$ on the whole of $\BH$.

\paragraph{The eigenstate-eigenvalue link.} We will now show that the arrow $\dasB{\A}$ constructed from a self-adjoint operator $\A\in\N_{\sa}$ preserves the eigenstate-eigenvalue link in a suitable sense. We employ the relation between ultrafilters in $\PV$ of the form $F_{\P_\psi}$ (see equation \eq{F_P_psi}) and maximal filters in $\PN$ established in the previous paragraph, but `read it backwards'.

Let $\A$ be some self-adjoint operator, and let $\psi$ be an eigenstate of $\A$ with eigenvalue $a$. Let $V$ be a context that contains $\A$, and let $\P_\psi$ be the projection determined by $\psi$, i.e., the projection onto the one-dimensional subspace (ray) $\bbC\psi$ of Hilbert space. Then $\P_\psi\in V$ from the spectral theorem. Consider the maximal filter
\begin{equation}
			F:=\{\Q\in\PN \mid \Q\geq\P_\psi\}
\end{equation}
in $\PN$ determined by $\P_\psi$. Clearly, $F\cap\PV=\{\Q\in\PV \mid \Q\geq\P_\psi\}$ is an ultrafilter in $\PV$, namely the ultrafilter $F_{\P_\psi}$ from equation \eq{F_P_psi}, and $F=\mc C_{\N}(F_{\P_\psi})$. Let $\ld_\psi$ be the element of $\Sig_V$ (i.e., local state of $V$) determined by $F_{\P_\psi}$. The state $\ld_\psi$ is nothing but the restriction of the vector state $\bra\psi\_\ket\psi$ (on $\N$) to the context $V$, and hence 
\begin{equation}
			\bra\psi\A\ket\psi=\ld_\psi(\A)=a.
\end{equation}
Since $\A\in V_{\sa}$, we have $\dastoi{V}{A}=\dastoo{V}{A}=\A$. This implies that 
\begin{align}			\nonumber
			\dasB{\A}_V(\ld_\psi)(V) &= [g_{\dastoi{V}{A}}(\ld_\psi),f_{\dastoo{V}{A}}(\ld_\psi)]\\			\nonumber
			&= [g_\A(\ld_\psi),f_\A(\ld_\psi)]\\
			&= [\fu{A}(\ld_\psi),\fu{A}(\ld_\psi)]\\		\nonumber
			&= [\ld_\psi(\A),\ld_\psi(\A)]\\		\nonumber
			&= [a,a],
\end{align}
i.e., the arrow $\dasB{\A}$ delivers the (interval only containing the) eigenvalue $a$ for an eigenstate $\psi$ (resp. the corresponding $\ld_\psi$) at a context $V$ which actually contains $\A$. In this sense, the eigenstate-eigenvalue link is preserved, and locally at $V$ the value of $\A$ actually becomes a single real number as expected.

\begin{example}			\label{Ex_ValueOfS_z}
We return to our example, the spin-$1$ system. It is described by the algebra $\mc{B}(\bbC^3)$ of bounded operators on the three-dimensional Hilbert space $\bbC^3$. In particular, the spin-$z$ operator is given by
\begin{equation}
			\hat S_z=\frac{1}{\sqrt{2}}
			\left(\begin{array}
						[c]{ccc}%
						1 & 0 & 0\\
						0 & 0 & 0\\
						0 & 0 & -1
			\end{array}\right).
\end{equation}
This is the matrix expression for $\hat S_z$ with respect to the basis $e_1,e_2,e_3$ of eigenvectors of $\hat S_z$. Let $\P_1,\P_2,\P_3$ be the corresponding rank-$1$ projections onto the eigenspaces. The algebra $V_{\hat S_z}:=\{\P_1,\P_2,P_3\}''$ generated by these projections is a maximal abelian subalgebra of $\mc{B}(\bbC^3)$, and it is the only maximal abelian subalgebra containing $\hat S_z$.

Let us now consider how the approximations $\dastoi{V}{S_z}$ and $\dastoo{V}{S_z}$ of $\hat S_z$ to other contexts $V$ look like. In order to do so, we first determine the antonymous and the observable function of $\hat S_z$. The spectral family of $\hat S_z$ is given by
\begin{equation}			\label{SpecFamOfSz}
			\hat E^{\hat S_z}_\ld=\left\{
			\begin{tabular}
						[c]{ll}%
						$\hat 0$ & if $\ld<-\frac{1}{\sqrt 2}$\\
						$\P_3$ & if $-\frac{1}{\sqrt 2}\leq\ld<0$\\
						$\P_3+\P_2$ & if $0\leq\ld<\frac{1}{\sqrt 2}$\\
						$\hat 1$ & if $\ld\geq \frac{1}{\sqrt 2}$.
			\end{tabular}
			\ \right.
\end{equation}
The antonymous function $g_{\hat S_z}$ of $\hat S_z$ is
\begin{align}
			g_{\hat S_z}:\mc{F}(\mc{B}(\bbC^3)) &\longrightarrow \spec{\hat S_z}\\    \nonumber
			F &\longmapsto \on{sup}\{r\in\bbR \mid \hat 1-\hat E^{\hat S_z}_r\in F\},
\end{align}
and the observable function is
\begin{align}
			f_{\hat S_z}:\mc{F}(\mc{B}(\bbC^3)) &\longrightarrow \spec{\hat S_z}\\    \nonumber
			F &\longmapsto \on{inf}\{r\in\bbR \mid \hat E^{\hat S_z}_r\in F\}.
\end{align}
Let $\Q_1,\Q_2,\Q_3$ be three pairwise orthogonal rank-$1$ projections, and let $V=\{\Q_1,\Q_2,\Q_3\}''\in\mc V (\mc B (\bbC^3))$ be the maximal context determined by them. Then the Gel'fand spectrum $\Sig_V$ has three elements; let us denote them $\ld_1,\ld_2$ and $\ld_3$ (where $\ld_i(\Q_j)=\delta_{ij}$). Clearly, $\Q_i$ is the smallest projection in $V$ such that $\ld_i(\Q_i)=1\;(i=1,2,3)$. Each $\ld_i$ defines a maximal filter in the projection lattice $\PV$ of $V$:
\begin{equation}
			F_{\ld_i}=\{\Q\in\PV \mid \Q\geq\Q_i\}.
\end{equation}
The Gel'fand transform $\fu{\dastoi{V}{S_z}}$ of $\dastoi{V}{S_z}$ is a function from the Gel'fand spectrum $\Sig_V$ of $V$ to the spectrum of $\hat S_z$. Equation \eq{EqInnDasFromAnton} shows how to calculate this function from the antonymous function $g_{\hat S_z}$:
\begin{equation}
				\fu{\dastoi{V}{S_z}}(\ld_i)=g_{\hat S_z}(\mc C_{\mc B (\bbC^3)}(F_{\ld_i})),
\end{equation}
where $\mc C_{\mc B (\bbC^3)}(F_{\ld_i})$ is the cone over the filter $F_{\ld_i}$, i.e.,
\begin{align}			\nonumber
			\mc C_{\mc B (\bbC^3)}(F_{\ld_i}) 
			&= \{\hat R\in\mc P (\mc B(\bbC^3)) \mid \exists \Q\in F_{\ld_i}:\hat R\geq\Q\}\\
			&= \{\hat R\in\mc P (\mc B(\bbC^3)) \mid \hat R\geq\Q_i\}.
\end{align}
Now it is easy to actually calculate $\fu{\dastoi{V}{S_z}}$: for all $\ld_i\in\Sig_V$,
\begin{align}			\label{deltai(Sz)FromAnton}	\nonumber
			\fu{\dastoi{V}{S_z}}(\ld_i)
			&= \on{sup}\{r\in\bbR \mid \hat 1-\hat E^{\hat S_z}_r\in\mc C_{\mc B (\bbC^3)}(F_{\ld_i})\}\\
			&= \on{sup}\{r\in\bbR \mid \hat 1-\hat E^{\hat S_z}_r\geq\Q_i\}.
\end{align}
Analogously, we obtain $\fu{\dastoo{V}{S_z}}$: for all $\ld_i$,
\begin{align}			\label{deltao(Sz)FromAnton}	\nonumber
			\fu{\dastoo{V}{S_z}}(\ld_i)
			&= \on{inf}\{r\in\bbR \mid \hat E^{\hat S_z}_r\in\mc C_{\mc B (\bbC^3)}(F_{\ld_i})\}\\
			&= \on{inf}\{r\in\bbR \mid \hat E^{\hat S_z}_r\geq\Q_i\}.
\end{align}
Using the expression \eq{SpecFamOfSz} for the spectral family of $\hat S_z$, we can see directly that the values $\fu{\dastoi{V}{S_z}}(\ld_i),\fu{\dastoo{V}{S_z}}(\ld_i)$ lie in the spectrum of $\hat S_z$ as expected. A little less obviously, $\fu{\dastoi{V}{S_z}}(\ld_i)\leq\fu{\dastoo{V}{S_z}}(\ld_i)$, so we can think of the pair of values as an interval $[\fu{\dastoi{V}{S_z}}(\ld_i),\fu{\dastoo{V}{S_z}}(\ld_i)]$.

We had assumed that $V=\{\Q_1,\Q_2,\Q_2\}''$ is a \emph{maximal} abelian subalgebra, but this is no actual restriction. A non-maximal subalgebra $V'$ has a Gel'fand spectrum $\Sig_{V'}$ consisting of two elements. All arguments work analogously. The important point is that for each element $\ld$ of the Gel'fand spectrum, there is a unique projection $\Q$ in $V'$ corresponding to $\ld$, given as the smallest projection in $V'$ such that $\ld(\Q)=1$.

This means that we can now write down explicitly the natural transformation $\dasB{S_z}:\Sig\rightarrow\Rlr$, the arrow in the presheaf topos representing the physical quantity `spin in $z$-direction', initially given by the self-adjoint operator $\hat S_z$. For each context $V\in\mc V (\mc B (\bbC^3))$, we have a function
\begin{align}			\label{dasB(S_z)_V}
			\dasB{S_z}_V:\Sig_V &\longrightarrow \Rlr_V\\			\nonumber
			\ld &\longmapsto (\mu_{\ld},\nu_{\ld}),
\end{align}
compare \eq{dasB(A)_V}. According to \eq{mu_ld}, $\mu_\ld:\downarrow\!\!V\rightarrow\bbR$ is given as
\begin{align}
				\mu_\ld:\downarrow\!\!V &\longrightarrow \bbR\\			 \nonumber
				V' &\longmapsto \ld(\delta^i(\hat S_z)_{V'})=\fu{\delta^i(\hat S_z)_{V'}}(\ld).
\end{align}
Let $\Q_{V'}\in\mc P (V')$ be the projection in $V'$ corresponding to $\ld$, that is, the smallest projection in $V'$ such that $\ld(\Q_{V'})=1$. Note that this projection $\Q_{V'}$ depends on $V'$. For each $V'\in\downarrow\!\!V$, the element $\ld$ (originally an element of the Gel'fand spectrum $\Sig_V$) is considered as an element of $\Sig_{V'}$, given by the restriction $\ld|_{V'}$ of the original $\ld$ to the smaller algebra $V'$. Then \eq{deltai(Sz)FromAnton} implies, for all $V'\in\downarrow\!\!V$,
\begin{equation}
				\mu_\ld(V')=\on{sup}\{r\in\bbR \mid \hat 1-\hat E^{\hat S_z}\geq\Q_{V'}\}.
\end{equation}
From \eq{nu_ld}, we obtain
\begin{align}
				\nu_\ld:\downarrow\!\!V &\longrightarrow \bbR\\			\nonumber
				V' &\longmapsto \ld(\delta^o(\hat S_z)_{V'})=\fu{\delta^o(\hat S_z)_{V'}}(\ld),
\end{align}
and with \eq{deltao(Sz)FromAnton}, we get for all $V'\in\downarrow\!\!V$,
\begin{equation}
				\nu_\ld(V')=\on{inf}\{r\in\bbR \mid \hat E^{\hat S_z}\geq\Q_{V'}\}.
\end{equation}

We finally want to calculate the `value' $\dasB{\hat S_z}(\ps\wpsi)$ of the physical quantity `spin in $z$-direction' in the (pseudo-)state $\ps\wpsi$. As described in section \ref{SubS_PhysQuantitiesAsArrows}, $\psi$ is a unit vector in the Hilbert space (resp. a vector state) and $\ps\wpsi=\ps{\delta(\P_\psi)}$ is the corresponding pseudo-state, a subobject of $\Sig$. The `value' $\dasB{\hat S_z}(\ps\wpsi)$ is given at $V\in\mc V (\mc B (\bbC^3))$ as
\begin{equation}
            (\dasB{\hat S_z}(\ps\wpsi))_V=\dasB{\hat S_z}_V(\ps\wpsi_V)
            =\{\dasB{\hat S_z}_V(\ld) \mid \ld\in\ps\wpsi_V\}.
\end{equation}
Here, $\ps\wpsi_V=\alpha(\delta^o(\P_\psi)_V)=S_{\delta^o(\P_\psi)_V}$ is a (clopen) subset of $\Sig_V$ (compare \eq{alpha}). This means that for each context $V$, the state determines a collection $\ps\wpsi_V$ of elements of the Gel'fand spectrum $\Sig_V$ of $V$. We then evaluate the component $\dasB{\hat S_z}_V$ of the arrow/natural transformation representing spin-$z$, given by \eq{dasB(S_z)_V} (and subsequent equations), on all the $\ld\in\ps\wpsi_V$ to obtain the component at $V$ of the `value' $\dasB{\hat S_z}(\ps\wpsi)$. Each $\ld\in\ps\wpsi$ gives a sequence of intervals; one interval for each $V'\in\downarrow\!\!V$ such that if $V''\subset V'$, the interval at $V''$ contains the interval at $V'$.
\end{example}

The example can easily be generalised to other operators and higher dimensions. Other finite-dimensional Hilbert spaces present no further conceptual difficulty at all. Of course, infinite-dimensional Hilbert spaces bring a host of new technical challenges, but the main tools used in the calculation, the antonymous and observable functions ($g_\A$ resp. $f_\A$), are still available. Since they encode the approximations in the spectral order of an operator $\A$ to \emph{all} contexts $V\in\VN$, the natural transformation $\dasB{\A}$ corresponding to a self-adjoint operator $\A$ can be written down efficiently, without the need to actually calculate the approximations to all contexts separately.

\section{Conclusion}	\label{S6}
We have shown how daseinisation relates central aspects of the standard Hilbert space formalism of quantum theory to the topos formalism. Daseinisation of projections gives subobjects of the spectral presheaf $\Sig$. These subobjects form a Heyting algebra, and every pure state allows to assign truth-values to all propositions. The resulting new form of quantum logic is contextual, multi-valued and intuitionistic. For more details on the logical aspects, see \cite{DI08,Doe07b,Doe09}.

The daseinisation of self-adjoint operators gives arrows from the state object $\Sig$ to the presheaf $\Rlr$ of `values'. Of course, the `value' $\dasB{\A}(\ps\wpsi)$ of a physical quantity in a pseudo-state $\ps\wpsi$ is considerably more complicated than the value of a physical quantity in classical physics, which is just a real number. The main point, though, is that in the topos approach all physical quantities \emph{do} have (generalised) values in any given state -- something that clearly is not the case in ordinary quantum mechanics. Moreover, we have shown in Example \ref{Ex_ValueOfS_z} how to calculate the `value' $\dasB{\hat S_z}(\ps\wpsi)$.

Recently, Heunen, Landsman and Spitters suggested a closely related scheme using topoi in quantum theory (see \cite{HLS09} as well as their contribution to this volume and references therein). All the basic ingredients are the same: a quantum system is described by an algebra of physical quantities, in their case a Rickart $C^*$-algebra, associated with this algebra is a spectral object whose subobjects represent propositions, and physical quantities are represented by arrows in a topos associated with the quantum system. The choice of topos is very similar to ours: as the base category, one considers all abelian subalgebras (i.e., contexts) of the algebra of physical quantities and orders them partially under inclusion. In our scheme, we choose the topos to be contravariant, $\Set$-valued functors (called \emph{presheaves}) over the context category, while Heunen et al. choose \emph{co}variant functors. The use of covariant functors allows the construction of the spectral object as the topos-internal Gel'fand spectrum of a topos-internal abelian $C^*$-algebra canonically defined from the external non-abelian algebra of physical quantities. 

Despite the similarities, there are some important conceptual and interpretational differences between the original contravariant approach and the covariant approach. These differences and their physical consequences will be discussed in a forthcoming article \cite{Doe09b}.

Summing up, the topos approach provides a reformulation of quantum theory in a way that seemed impossible up to now due to powerful no-go theorems like the Kochen-Specker theorem. In spite of the Kochen-Specker theorem, there is a suitable notion of state `space' for a quantum system in analogy to classical physics: the spectral presheaf. 

\paragraph{Acknowledgements.} I would like to thank Chris Isham for the great collaboration and his constant support. I am very grateful to Hans Halvorson for organising the very enjoyable \textsl{Deep Beauty} meeting.


\begin{thebibliography}{999999}

\bibitem{BvN36} G.~Birkhoff and J.~von Neumann.
\newblock The logic of quantum mechanics.
\newblock {\em Ann.\ Math.} {\bf 37}, 823--843 (1936).

\bibitem{DCG02} M.L.~Dalla Chiara, R.~Giuntini.
\newblock Quantum logics. In {\em Handbook of Philosophical Logic}, vol.\ VI,
\newblock eds. G. Gabbay and F. Guenthner.
\newblock Kluwer, Dordrecht, 129--228 (2002).

\bibitem{Doe05} A.~D\"{o}ring.
\newblock Kochen-Specker theorem for von Neumann algebras.
\newblock {\em Int.\ Jour.\ Theor.\ Phys.} {\bf 44}, 139--160 (2005).
\newblock arXiv:quant-ph/0408106

\bibitem{Doe07b} A.~D\"{o}ring.
\newblock Topos theory and `neo-realist' quantum theory.
\newblock In \textit{Quantum Field Theory, Competitive Models}, eds. B. Fauser, J. Tolksdorf, E. Zeidler. 
\newblock Birkh\"auser, Basel, Boston, Berlin (2009).
\newblock arXiv:0712.4003.

\bibitem{Doe08} A.~D\"{oring}.
\newblock Quantum States and Measures on the Spectral Presheaf.
\newblock \emph{Adv. Sci. Lett.} \textbf{2}, Number 2 (Special Issue on ``Quantum Gravity, Cosmology and Black Holes'', ed. M. Bojowald), 291--301 (2009). arXiv:0809.4847

\bibitem{Doe09} A.~D\"{o}ring.
\newblock Topos Quantum Logic and Mixed States.
\newblock To appear in \textit{Electronic Notes in Theoretical Computer Science} (6th Workshop on Quantum Physics and Logic, QPL VI, Oxford, 8.--9. April 2009), eds. B. Coecke, P. Panangaden, P. Selinger (2010)

\bibitem{Doe09b} A.~D\"{o}ring.
\newblock{Algebraic quantum theory in a topos: a comparison.}
\newblock In preparation (2010).

\bibitem{DI(1)} A.~D\"{o}ring, and C.J. Isham.
\newblock A topos foundation for theories of physics:
    {I.} Formal languages for physics.
\newblock \emph{J. Math. Phys.} \textbf{49}, Issue 5 (2008).
\newblock arXiv:quant-ph/0703060.

\bibitem{DI(2)} A.~D\"{o}ring, and C.J. Isham.
\newblock A topos foundation for theories of physics:
    {II.} Daseinisation and the liberation of quantum theory.
\newblock \emph{J. Math. Phys.} \textbf{49}, Issue 5 (2008).
\newblock arXiv:quant-ph/0703062.

\bibitem{DI(3)} A.~D\"{o}ring, and C.J. Isham.
\newblock A topos foundation for theories of physics:
    {III.} Quantum theory and the representation of
    physical quantities with arrows \mbox{$\dasB{A}:\Sig\rightarrow\ps{{\mathbb{R}}^\succeq}$}.
\newblock \emph{J. Math. Phys.} \textbf{49}, Issue 5 (2008).
\newblock arXiv:quant-ph/0703064.

\bibitem{DI(4)} A.~D\"{o}ring, and C.J. Isham.
\newblock A topos foundation for theories of physics:
            {IV.} Categories of systems.
\newblock \emph{J. Math. Phys.} \textbf{49}, Issue 5 (2008).
\newblock arXiv:quant-ph/0703066.

\bibitem{DI08} A.~D\"{o}ring, C.J.~Isham.
\newblock `What is a Thing?': Topos Theory in the Foundations of Physics. To appear in \textit{New Structures in Physics}, ed. B. Coecke, Springer Lecture Notes in Physics, Springer, Berlin, Heidelberg, New York (2010).
\newblock arXiv:0803.0417.

\bibitem{deG04}H.F.~de~Groote.
\newblock On a canonical lattice structure on the effect algebra of a von Neumann algebra.
\newblock arXiv:math-ph/0410.018v2 (2004).

\bibitem{deG05c}H.F.~de~Groote.
\newblock Observables I: Stone Spectra.
\newblock arXiv:math-ph/0509.020 (2005).

\bibitem{deG05b}H.F.~de~Groote.
\newblock Observables II: Quantum Observables.
\newblock arXiv:math-ph/0509.075 (2005).

\bibitem{deG07}H.F.~de~Groote.
\newblock Observables IV: The Presheaf Perspective.
\newblock arXiv:0708.0677 [math-ph] (2007).

\bibitem{Emch84}G. G. Emch.
\newblock Mathematical and conceptual foundations of 20th-century physics.
\newblock North Holland, Amsterdam, New York, Oxford (1984).

\bibitem{Flo08}C.~Flori.
\newblock A Topos Formulation of Consistent Histories.
\newblock arXiv:0812.1290.

\bibitem {HLS09} C.~Heunen, N.P.~Landsman, B. Spitters. 
\newblock{A topos for algebraic quantum theory.}
\newblock{\textit{Commun. Math. Phys.} \textbf{291}(1), 63--110 (2009)}.

\bibitem {Isham97} C.J.~Isham.
\newblock{Topos theory and consistent histories:  The internal
logic of the set of all consistent sets}.
\newblock{\em Int.\ J.\ Theor.\ Phys.} \textbf{36}, 785--814 (1997).
\newblock arXiv:gr-qc/9607069

\bibitem {Ish10} C.J.~Isham.
\newblock Topos Methods in the Foundations of Physics.
\newblock To appear in \textit{Deep Beauty}, ed. Hans Halvorson, Cambridge University Press (2010).

\bibitem{IB98} C.J.~Isham and J.~Butterfield.
\newblock A topos perspective on the {K}ochen-{S}pecker theorem:
            {I.} {Q}uantum states as generalised
valuations. \newblock {\em Int.\ J.\ Theor.\ Phys.} \textbf{37},
2669--2733 (1998).
\newblock arXiv:quant-ph/9803055v4

\bibitem{IB99} C.J.~Isham and J.~Butterfield.
\newblock A topos perspective on the {K}ochen-{S}pecker theorem:
    {II.} {C}onceptual aspects, and classical analogues.
\newblock  {\em Int.\ J.\ Theor.\ Phys.} \textbf{38}, 827--859 (1999).
\newblock arXiv:quant-ph/9808067v2

\bibitem{IB00} C.J.~Isham, J.~Hamilton and J.~Butterfield.
\newblock A topos perspective on the {K}ochen-{S}pecker theorem:
    {III.} {V}on {N}eumann algebras as the base category.
\newblock  {\em Int.\ J.\ Theor.\ Phys.} \textbf{39}, 1413-1436 (2000).
\newblock arXiv:quant-ph/9911020

\bibitem{IB02} C.J.~Isham and J.~Butterfield.
\newblock {A topos perspective on the Kochen-Specker theorem:
        {IV.} {I}nterval valuations}.
\newblock {\em Int.\ J.\ Theor.\ Phys} {\bf 41}, 613--639 (2002).
\newblock arXiv:quant-ph/0107123

\bibitem{KR83a} R.V.~Kadison and J.R.~Ringrose.
\newblock {\em Fundamentals of the Theory of Operator Algebras,
            Volume 1: Elementary Theory}.
            \newblock Academic Press, New York (1983).

\bibitem{McL98} S.~Mac{L}ane.
\newblock {\em Categories for the Working Mathematician, Second Edition}.
\newblock Springer, New York, Berlin, Heidelberg (1998).

\bibitem{Ols71} M.P.~Olson.
\newblock The Selfadjoint Operators of a von
Neumann Algebra form a Conditionally Complete Lattice.
\newblock {\em Proc. of the AMS}\;\;{\bf 28}, 537--544 (1971).

\bibitem{Tak79} M.~Takesaki.
\newblock Theory of Operator Algebras I.
\newblock Springer, Berlin, New York (1979).

\end{thebibliography}
\end{document}